 \documentclass[reprint,amsmath,amssymb,aps,pre]{revtex4-1}

\usepackage{amsmath,amssymb,amsfonts,amstext,graphics,graphicx,subfigure,dcolumn,bm}
\usepackage{setspace}
\usepackage{mathtools}
\usepackage[colorlinks]{hyperref}
\usepackage{color}

\usepackage{mathbbol}

  \usepackage{natbib}

\newtheorem{theorem}{Theorem}

\newtheorem{lemma}[theorem]{Lemma}

\newenvironment{proof}[1][Proof]{\noindent\textbf{#1.} }{\ \rule{0.5em}{0.5em}}

\def\Xint#1{\mathchoice
{\XXint\displaystyle\textstyle{#1}}%
{\XXint\textstyle\scriptstyle{#1}}%
{\XXint\scriptstyle\scriptscriptstyle{#1}}%
{\XXint\scriptscriptstyle\scriptscriptstyle{#1}}%
\!\!\int}
\def\XXint#1#2#3{{\setbox0=\hbox{$#1{#2#3}{\int}$ }
\vcenter{\hbox{$#2#3$ }}\kern-.5\wd0}}

\def\dashint{\Xint-}

\def\calc{\mathcal{C}}
\def\cald{\mathcal{D}}

\def\calf{\mathcal{F}}
\def\calg{\mathcal{G}}
\def\calh{\mathcal{H}}

\def\calj{\mathcal{J}}
\def\calk{\mathcal{K}}

\def\calt{\mathcal{T}}

\def\calw{\mathcal{W}}

\def\calz{\mathcal{Z}}

\def\C{\mathbb{C}}

\def\R{\mathbb{R}}
\def\C{\mathbb{C}}

\def\bq{\begin{equation}}
\def\eq{\end{equation}}
\def\bqy{\begin{eqnarray}}
\def\eqy{\end{eqnarray}}

\def\bal#1\eal{\begin{align}#1\end{align}}

\def\al{\alpha}
\def\be{\beta}

\def\de{\delta}

\def\ep{\epsilon}

\def\et{\eta}
\def\ga{\gamma}
\def\Ga{\Gamma}

\def\ka{\kappa}
\def\la{\lambda}
\def\La{\Lambda}

\def\om{\omega}

\def\si{\sigma}
\def\Si{\Sigma}

\def\bfd{\mathbf{d}}

\def\bfM{\mathbf{M}}

\def\bfL{\mathbf{L}}

\def\bfx{\mathbf{x}}
\def\bfR{\mathbf{R}}
\def\bfW{\mathbf{W}}
\def\bfX{\mathbf{X}}
\def\bfY{\mathbf{Y}}
\def\bfZ{\mathbf{Z}}
\def\bfv{\mathbf{v}}

\def\p{\partial}

\def\andq{\quad\mathrm{and}\quad}

\def\ncr{\nonumber\\}
\def\nn{\nonumber}

\def\qqquad{\quad\quad\quad}

\newcommand{\KN}{\mathbin{\bigcirc\mspace{-15mu}\wedge\mspace{3mu}}}

\begin{document}

\title{An inclusive curvature-like framework for describing dissipation: \\
 metriplectic 4-bracket dynamics}
\author{Philip J.~Morrison}
\email{morrison@physics.utexas.edu}
\author{Michael H.~Updike}
\email{michaelupdike@utexas.edu}
\affiliation{Department of Physics and Institute for Fusion Studies, 
The University of Texas at Austin, Austin, TX, 78712, USA}
\date{\today}

\begin{abstract}

\bigskip

An inclusive framework for joined Hamiltonian and dissipative dynamical systems, which preserve energy and produce entropy,  is given.  The dissipative dynamics of the framework is based on the {\it metriplectic 4-bracket}, a quantity like the Poisson bracket defined on phase space functions, but unlike the Poisson bracket has four slots with symmetries and properties motivated by Riemannian curvature.  Metriplectic 4-bracket dynamics is generating using two generators, the Hamiltonian and the entropy, with the   entropy being a Casimir of the Hamiltonian  part of the system.  The formalism includes all known previous binary bracket theories for dissipation or relaxation as special cases.  Rich geometrical significance of the  formalism and methods for constructing metriplectic 4-brackets are explored.  Many  examples of both finite and infinite dimensions are given.

\end{abstract}

\maketitle

\begin{widetext}

\tableofcontents

\end{widetext}

\section{Introduction}

\subsection{Description}

Various proposals have been suggested for categorizing or placing into a general formalism dissipative effects added to Lagrangian or Hamiltonian dynamical systems.  An early example is the formalism of  Rayleigh \cite{rayleigh77} who proposed a generalization of Lagrangian mechanics, by adding a term, the Rayleigh dissipation function, to Lagrange's equations of motion. However, here we follow on  the early 1980's formalisms based on adding to generalizations of the Poisson bracket (see e.g.\  \cite{pjm82,pjm98}) a  bilinear bracket which, akin to the Poisson bracket, is defined on  phase space functions.  These bracket  formalisms  were proposed in \cite{pjmK82,pjmH84,kauf84,pjm84,pjm84b,pjm86} for describing dynamics with dissipation in finite-dimensional systems,  fluid mechanics,  plasma models, and kinetic theories.  In this paper  we present an encompassing geometric formulation in terms of a quantity called the metriplectic   4-bracket,  which like the Poisson bracket is defined on phase space functions, but  has properties motivated by the  Riemann curvature  tensor, which subsumes ideas from the above publications as well as the the double bracket formalism of \cite{ brockett,vallis,carn90,shep90,pjmF11}.  

A variety of metriplectic 4-bracket examples of both finite and infinite dimension will be described here, the former by reduction  or mechanical  modeling  while  the latter come from fluid mechanics, plasma dynamics, and kinetic theory. However, we note,  there are gradient-like systems for describing dissipative dynamics in other  areas,  such as the Cahn-Hilliard equation \cite{ch58} that describes phase separation for binary fluids,  the gradient structure of porous medium dynamics \cite{otto01}, and even the Ricci Flows \cite{ham82,ch58} instrumental in the proof  of the Poincar\`e conjecture on $S^3$ \cite{perelman02}.  We have  found that the metriplectic 4-bracket formalism encompasses all  cases we have examined, but additional examples along these lines are reserved for future  publications.   Indeed, the metriplectic 4-bracket  formalism is most inclusive.

Our starting point is to consider  finite-dimensional systems with a  phase space $\calz$ being an $n$-dimensional manifold, on which a multibracket of the form $(f,g,h, \dots)$ is defined on smooth real-valued functions $f,g,h,\dots\in C^{\infty}(\calz)$, so that we have  
\[
(\, \cdot\, ,\, \cdot\,,\,\cdot\,,\, \cdot\,  \dots)\colon C^{\infty}\!(\calz)\times C^{\infty}\!(\calz)\times C^{\infty}\!(\calz)\times\dots\rightarrow C^{\infty}\!(\calz)\,.
\]
Examples of multilinear brackets of this form include Albeggiani's Poisson bracket of the $nth$ order (see \cite{whittaker37} p.~337), the well-known and oft cited Nambu bracket \cite{nambu73} (which was predated by Albeggiani), and the Lie algebra generalization of Nambu given in \cite{pjmB91}.  An $n$-bracket of this type generates nondissipative dynamics upon specification of $n-1$ ``Hamiltonians" as follows:
\[
\dot{o} =(o,H_1,H_2, \dots, H_{n-1})\,,
\]
where $H_1,H_2, \dots, H_{n-1}$ are the Hamiltonians and the observable $o$ is any dynamical variable.  Usual Poisson brackets, canonical or noncanonical, correspond to the case of a bilinear  bracket with a single Hamiltonian. 

For infinite-dimensional systems, field theories, multibrackets are  defined on functionals, maps of real-valued functions defined on some function space,
that contains the dynamical variables  defined on $\cald$,  the continuum label space of the field theory.  Here, usually out of necessity, since the main systems of interest are nonlinear partial or partial-integro differential equations, we operate on a formal level and assume whatever is necessary for our 
operations to exist. 

While the finite- and infinite-dimensional  multibrackets described above generate nondissipative dynamics,  the purpose of metriplectic dynamics, as developed in \cite{pjm84,pjm84b,pjm86,pjm09a,Materassi12,pjmBR13,pjmM18,pjmC20}, is to place the first and second laws  of thermodynamics into  a dynamical systems setting using  both a noncanonical (degenerate) Poisson bracket  and a symmetric bilinear  dissipative bracket, which together  produce a combination of a Hamiltonian system with a gradient system associated with a degenerate metric-like  tensor,  providing a Lyapunov function by construction.   Consequently, the two main functions of thermodynamics, an energy  or Hamiltonian $H$ and an entropy $S$  play important roles in the dynamics.  It was for this reason that a parent  3-bracket $(f,g;h)$  was given  in \cite{pjm84}  that reduces to  the dissipative bilinear bracket of metriplectic dynamics, where the $H$ dependence was designed as a  projector in order to maintain energy conservation.  Additional projector examples were given in \cite{pjm86} and subsequent work and a general dissipative multibracket that can preserve additional variables were provided in \cite{pjmBR13}. 

Given that dissipation is governed by a kind of  gradient system with a kind of metric tensor,  the idea that a  curvature tensor like object could be associated with dissipation was put forth in \cite{pjm86}.  So, given a  curvature-like  tensor, with two important functions associated with the dynamics, viz.\  $H$ and $S$,  with the former conserved and the latter produced, we are led to the idea of a 4-bracket of the form $(f,k;g,h)$ with symmetries and properties consistent with those possessed  by a fully  contravariant  curvature 4-tensor.  It turns out that this is a general point of view that encompasses a wide variety of dissipative dynamics formalisms, including all the previous bilinear brackets  for dissipation. We note  that the  formalism called GENERIC (rooted in \cite{grm84} but developed in \cite{og97} and subsequent works)  does not fit directly into the framework given here.  This is because GENERIC (often used improperly to mean the prior metriplectic dynamics) is not bilinear and lacks the requisite symmetry in its binary dissipative bracket, a symmetry that would be induced, as we will see, by an  underlying metriplectic 4-bracket.   However, a procedure is given for turning GENERIC brackets into metriplectic brackets and thereby fitting them into the present theory. 

\subsection{Overview}

The paper is organized as follows.  Section \ref{sec:finite} is about finite-dimensional systems.  Here, in Secs.\ \ref{ssec:poisson} and \ref{ssec:MPdynamics},  we review Poisson bracket and metriplectic dynamics, respectively and  establish notation that is to be used.  Section \ref{ssec:mp4bktd} contains the main new formalism of the  paper, dissipative dynamics generated by a metriplectic 4-bracket. The basic  notion is described in Sec.\ \ref{sssec:curvature4bkt}, where the {\it minimal metriplectic} properties of the 4-bracket are introduced, and in Sec.\ \ref{ssec:curvature4bkt-metri}  it is shown how a metriplectic 4-bracket with these properties generates dissipative dynamics consistent with the laws of thermodynamics.  Here we see how the metriplectic 2-bracket of \cite{pjm84,pjm86} emerges from the formalism.  In Sec.\  \ref{sssec:Lie-M} we describe a  geometric setting for manifolds with both Poissonian and Riemannian structure, which  we call {\it Lie-metriplectic}  {manifolds}.  Section  \ref{ssec:tremoval}  shows that there is,  in a sense that we define, a unique torsion-free metriplectic 4-bracket.  

Section  \ref{ssec:special4bkt}  describes paths for constructing metriplectic 4-brackets.  In Sec.\ \ref{sssec:LCforms} we see how they emerge from  manifolds with Riemannian structure, affine or Levi-Civita, showing there is a large class of possibilities.  In Sec.\ \ref{ssec:KN} the Kulkarni-Nomizu product  is adapted for out purposes -- it provides a way for  building-in the  requisite symmetries given two symmetric bivector fields of one's choosing. In Sec.\ \ref{sssec:LAMbkts} we describe Lie algebra based metriplectic 4-brackets, a formalism akin to   the Lie-Poisson manifold construction, with a special pure Lie algebra  case based on the  Cartan-Killing metric. Section \ref{sssec:poisson-connect}  uses a Poisson bracket induced connection (see  \cite{holonomy}) to describe a class of metriplectic 4-brackets with rich geometry. 

In Sec.\  \ref{ssec:reductionsF} we show how the  metriplectic 4-bracket formalism subsumes previous  binary  brackets.   In Sec.\ \ref{sssec:KM} we see how it reduces to the Kaufman-Morrison bracket,  while in Sec.\ \ref{sssec:double} we explore how it  relates to the double bracket formalism.  Finally, in this subsection, in Sec.\ \ref{ssec:GM}, we examine how GENERIC  can be linearized  and symmetrized, and then emerging naturally from the  metriplectic 4-bracket formalism. 

Section  \ref{ssec:finiteEX} contains a collection of finite-dimensional examples, beginning in \ref{sssec:Frb}  with the ubiquitous free rigid body  followed  in Sec.\ \ref{sssec:KKe} by the Kida  vortex, another Lie algebra based example. We conclude this section with Sec. \ref{sssec:other3D}, where other examples are mentioned. 

Section \ref{sec:infinite} describes the  leap from finite to infinite dimensions.    In Sec.\ \ref{ssec:field4bkts} we review noncanonical Hamiltonian field theory,  various  metriplectic and double bracket field theories, and  present a general form for the field-theoretic metriplectic 4-bracket. Section  \ref{ssec:recuctionsI}, is the infinite-dimensional version of  Sec.\ \ref{ssec:reductionsF}, where we show how metriplectic 4-bracket field theory subsumes previous theories. 
  
Infinite-dimensional examples are given in Secs.\ \ref{ssec:fluid} and   \ref{ssec:KT}.  Here we see how efficient it can be to construct  metriplectic 4-bracket field theories. In Secs.\   \ref{sssec:1+1} \ref{sssec:2+1} and  \ref{sssec:3+1} various 1-, 2-, and 3-dimensional fluid and plasma-like theories are developed, including one where the fluid helicity plays the role of entropy.  In Secs.\ \ref{sssec:landauCO} and \ref{sssec:GBoltz} concern kinetic-like theories.  The former giving a generalization of the Landau collision operator, that has proven useful for computing equilibria, while the latter concerns finding the metriplectic 4-bracket from which the  Boltzmann bracket of \cite{grm84} emerges. 

Finally, in  Sec.\ \ref{sec:conclusion}  we conclude.  Here we briefly summarize some main points of the paper and discuss the usefulness of the metriplectic 4-bracket formalism,  We discuss  how it can be used to develop `honest' models and we speculate about its usefulness  for structure preserving computation.


\section{Finite-dimensional metriplectic dynamical systems and the metriplectic 4-bracket}
\label{sec:finite}

 We consider  dynamical systems  a real phase space manifold $\calz$ of dimension $N$.  In a coordinate patch we denote a point of 
 $\calz$ by $z=(z^1,z^2,\dots, z^N)$ with usual tensorial notation.  For example, given a vector field $\bfZ\in \mathfrak{X}(\calz)$, where $\mathfrak{X}(\calz)$ denotes differentiable vector fields on our phase space manifold $\calz$,  we have the dynamical system, a set of autonomous first order differential equations, given by
 \bq
 \dot{z}^i= Z^i(z)\,,\qquad i=1,2,\dots, N\,,
 \eq
 with, as usual, $\cdot$ denoting time differentiation. 
 Such dynamics will be generated by bracket operations
 \bq
  C^\infty(\calz) \times C^\infty(\calz) \times \dots \mapsto C^\infty(\calz)
  \eq
defined on smooth functions $C^\infty(\calz)$. 
It is conventional to call  $C^\infty(\calz)$ the space of 0-forms, $\La^0(\calz)$, and we will use these expressions interchangeably.  Here, all quantities are assumed to be real-valued,  although we note that extensions from $\R$ to $\C$ are possible and remain to be fully explored.

\subsection{Poisson dynamics}
\label{ssec:poisson}

A phase space with Poisson manifold structure uses noncanonical Poisson brackets to generate flows.  (See \cite{weinstein,weinsteinE} for  important seminal work and \cite{sudarshan,pjm98} for a physicists perspective.).  Such a  binary operation, 
\bq
\{\,\cdot\,,\cdot\,\}\colon  C^\infty(\calz) \times C^\infty(\calz) \rightarrow C^\infty(\calz)\,,
\eq
in addition to being bilinear,  satisfies the following for all $f,g,h\in C^\infty(\calz)$: 
\bqy
&\mathrm{antisymmetry}:\quad &\{f,g\}=-\{g,f\}
\\
&\mathrm{Jacobi\ identity}:\quad & \{\{f,g\},h\}+ \{\{g,h\},f\}\} \nonumber\\
&  &  \hspace{1cm}+  \{\{h,f\},g\}\} =0\,,
\eqy
which provide  a Lie algebra realization on  $C^\infty(\calz)$, and
\bq
\mathrm{derivation}:\quad \{fg,h\}= f\{g,h\} + \{g,h\}f\,,
\label{derivation}
\eq
where $fg$ denotes pointwise multiplication of functions in $C^\infty(\calz)$.  The Leibniz derivation property of \eqref{derivation} assures that $\bfZ_f:=\{f, \,\cdot\,\} \in \mathfrak{X}(\calz)$ is a kind of Hamiltonian vector field.

For a Hamiltonian $H\in C^\infty(\calz)$ the equations of motion in a coordinate patch take the following form in tensorial notation:
\bq
\dot z^i=\{z^i,H\}=J^{ij}\frac{\p H}{\p z^j}\,,\qquad i,j=1,2,\dots, N\,,
\eq
where 
\bq
\{f,g\}= \frac{\p f}{\p z^i}J^{ij}\frac{\p g}{\p z^j}
\label{PBcoords}
\eq
and repeated index notation is used in the second equality. We will call the bivector field $(J^{ij})$ the Poisson tensor. 
We can write \eqref{PBcoords} in a two ways, namely 
\[
\{f,g\}= J(\bfd f, \bfd g) = \langle \bfd f,J\bfd g\rangle \,,
\]
{where for $f,g\in \La^0(\calz)$, the exterior derivative gives $\bfd f,\bfd g\in \La^1(\calz)$, the space of 1-forms, and in the second equality we have the duality pairing $\langle \ , \ \rangle$  between one-forms and vectors, with 
$J$ considered as a bundle map $J\colon T^*\calz \rightarrow T\calz$ satisfying $J^*=-J$.}

On symplectic manifolds, a special case of Poisson manifolds,  $N=2M$ and we may choose coordinates such that the Poisson tensor has the canonical form
\bq
J_c= 
\begin{pmatrix}
O_M& I_M\\
-I_M & O_M
\end{pmatrix}\,.
\label{Jc}
\eq
The choice of these coordinates reflects the usual splitting of $z$ into the canonical coordinates $(q,p)$. 

Unlike the nondegenerate canonical Poisson bracket with Poisson tensor \eqref{Jc},  where $\{f,C\}=0\ \forall f\, \Leftrightarrow f = \mathrm{constant}$ (a real  number),   on Poisson manifolds $\{f,C\}=0\ \forall f$ {is} satisfied by nontrivial Casimir invariants. Generally speaking, the level sets of these quantities foliate the Poisson manifold and dynamics is confined to the leaf tagged by an initial condition for {\it any} Hamiltonian function.  Such degenerate brackets, with Poisson tensor fields $(J^{ij})\neq (J_c^{ij})$, were called  noncanonical Poisson brackets in \cite{pjmG80,pjm80}. Casimir invariants play a special role as candidates for entropy functions in both the metriplectic formalism of Sec.~\ref{ssec:MPdynamics} and the curvature 4-bracket dynamics of Sec.~\ref{sssec:curvature4bkt}.  

Lie-Poisson brackets are special kind of noncanonical Poisson bracket that are associated with any Lie algebra $\mathfrak{g}$.   The natural phase space is actually the dual $\mathfrak{g}^*$.  For $f,g\in C^\infty(\mathfrak{g}^*)$,   $z\in \mathfrak{g}^*$, and $\bfd f\in\mathfrak{g}$,  the bracket has the form
\bqy
\{f,g\}&=& \langle z,[ \bfd f, \bfd  g]\rangle
\nonumber\\
&=&  \frac{\p f}{\p z^i} \, c^{ij}_{\ \, k} \, z_k\frac{\p g}{\p z^j}\,,
\label{LPbkt}
\eqy
where $[\,,\, ]$ is the Lie-bracket of $\mathfrak{g}$,   $i,j,k= 1,2, \dots, \mathrm{dim}\,  (\mathfrak{g})$  $z^i$ are coordinates for $\mathfrak{g}^*$,  and where $c^{ij}_{\ \,k}$ are the structure constants of  $\mathfrak{g}$.

Just as these Lie-Poisson brackets are special Poisson brackets associated with Lie algebras, we will find in Sec.~\ref{sssec:LAMbkts} that there are special metriplectic 4-brackets associated with Lie algebras.

\subsection{Metriplectic dynamics}
\label{ssec:MPdynamics}

As noted above, metriplectic dynamics emerged in the 1980s from the work 
\cite{pjmK82}, the adjacent papers \cite{kauf84,pjm84}, and \cite{pjm84b}, with the full set of axioms first appearing in \cite{pjm84,pjm84b}.  {The name metriplectic was introduced in \cite{pjm86} with the three basic axioms of  \cite{pjm84,pjm84b}   (given below) restated and several examples provided.}  The review here is  adapted from the more recent publication \cite{pjmBR13}.  As we will see, metriplectic dynamics nicely places the  First and Second Laws of  Thermodynamics into  a dynamical systems setting.

As above, a  \textit{metriplectic system} consists of a phase space manifold $\calz$, a Poisson bundle map $J\colon T^*\calz \rightarrow T\calz$, a bundle map  $G\colon T^*\calz \rightarrow T\calz$, and two functions $H, S \in C^{\infty}(\calz)$ with $H$ being the Hamiltonian (energy) and $S$ being the entropy.   
The dynamics will be defined in terms of binary brackets on functions $f,g,h\in C^\infty(\calz)$, which we assume have the following properties:
\begin{itemize}
 \item[(i)] $(f,g)\coloneqq \left\langle \mathbf{d}f, G\mathbf{d}g \right\rangle$ is a positive semidefinite symmetric bracket, i.e., $(\cdot\,,\cdot)$ is bilinear and symmetric, so $G^\ast = G$,
and $(f,f) \geq 0\  \forall\,  f \in C ^{\infty}(\calz)$; in coordinates the symmetric bracket has the form
\bq
(f,g)=  \frac{\p f}{\p z^i}G^{ij}\frac{\p g}{\p z^j}\,,
\label{SBcoords}
\eq
and   $G^{ij}=G^{ji}, \ \forall \,  i,j=1,2,\dots N$; 
\item[(ii)] $\{S, f\} = 0$ and $(H, f)=0\ \forall \, f \in C^{\infty}(\calz) 
\Longleftrightarrow J \mathbf{d}S = G\mathbf{d}H =0$; in coordinates this expresses the null space conditions
\bq
J^{ij}\frac{\p S}{\p z^j}\equiv 0 \quad \mathrm{and}\quad G^{ij}\frac{\p H}{\p z^j}\equiv 0\,.
\eq
\end{itemize}
The \textit{metriplectic dynamics} of any observable (dynamical variable) $o$ is given in terms of the two brackets by
\bqy
\label{metriplectic_dynamics_function}
\dot o&=& \{o, H+S\} + (o, H+S)
\nonumber\\
& =& \{o, H\} + (o, S)\,, \quad  \forall\, o \in C^{\infty}(\calz)\,.
\eqy

{We note that by scaling $S$ or the symmetric bracket, instead of using the single generator $H+S$ the dynamics can be written in terms of  $\calf=H-\calt  S$, where $\calt$ can  be interpreted as a global constant temperature and, consequently, $\calf$ is a Helmholtz free energy \cite{pjmC20}.}

In terms of coordinates in tensorial notation we have the ordinary differential equations
\begin{equation}
\label{metriplectic_dynamics}
\dot{z}^i = J^{ij}\frac{\p H}{\p z^j} + G^{ij}\frac{\p S}{\p z^j}\,.
\end{equation}
Geometrically, the  vector field $\bfZ_H\coloneqq J \mathbf{d}H \in \mathfrak{X}(\calz)$ expresses the Hamiltonian part of these equations, while $\bfY_S\coloneqq G\mathbf{d}S \in \mathfrak{X}(\calz)$ gives the dissipative part of the full metriplectic dynamics  of  \eqref{metriplectic_dynamics_function} or  \eqref{metriplectic_dynamics}. 

The name  metriplectic, as first  given  in  \cite{pjm86}, was chosen because these systems blend dissipative  and  Hamiltonian dynamics.  The dissipative part,  being  generated by the  symmetric  bracket,   is  a  degenerate  gradient flow determined by a  metric-like tensor $G^{ij}$  accounting for the `metri' of  metriplectic.  We will see in  this paper that there can also  be an actual  Riemannian metric, say $g$.  To distinguish the two, we will  refer  to the tensor $G^{ij}$  as   the  $G$-metric (even though it is  degenerate). Because  dynamics  in a Poisson manifold is symplectic on a Casimir leaf, this motivated the `plectic' of metriplectic.
 
The definition of metriplectic systems was {designed}  to have  three immediate and important consequences:
\begin{itemize}
\item[(i)] \quad \textit{Energy conservation -- First Law}:
\begin{equation}
\label{energy_conservation}
\dot{H} = \{H, H\} + (H, S) \equiv 0.
\end{equation}
\item[(ii)] \quad\textit{Entropy production -- Second Law}:
\bal
\label{entropy_production}
\dot{S}&= \{S, H\} + (S, S) 
\ncr
&= (S, S) \geq 0\,,
\eal
where the second equality follows because the entropy is a Casimir.  Here, in  line with thermodynamics, we have entropy  production; however, reversing the sign of the entropy gives a decreasing quantity as is typical for Lyapunov functions. 
\item[(3)] \quad\textit{Maximum entropy principle yields 
equilibria}:
Suppose that a point $z_*$ has any neighborhood $U$ such that for every
 point $z \neq z_* \in U$ such that $H(z)=H(z^*)$,  $S(z) < S(z_*)$.  Then, by the second law, $z_*$ is necessarily an equilibrium of the metriplectic dynamics. This is akin to the free 
energy extremization of thermodynamics, as noted  in 
\cite{pjm84b,pjm86} where it was  suggested that one can build in degeneracies associated with 
Hamiltonian ``dynamical constraints.'' For example, a good collision operator should conserve mass and momentum, in addition to energy.  (See also \cite{pjmBR13,mielke}.)  We will see that similar degeneracies can be {naturally} built into our {metriplectic} 4-bracket. 
\end{itemize}

Proving conventional nondegenerate  gradient flows achieve equilibrium states  has a  large literature dating to  \cite{Lojasiewicz,Polyak}.   Some results for  the degenerate flows  generated by the metric 4-brackets of 
Sec.\ \ref{ssec:mp4bktd} of the present  paper  are apparent, but  the nature of the level sets of $H$ and $S$ complicate matters, with multiple possible basins of attraction etc. Consequently, we leave this for future publication. (Some results will be included in \cite{pjmBKM23}.)

Although not treated in detail here, conservation of other invariants in addition to the Hamiltonian may be of interest.   Suppose that $I \in C ^{\infty}(\calz)$ is a  quantity conserved by the Hamiltonian part of the metriplectic dynamics, i.e., 
$\{I,H\} = 0$. Then, on an integral curve of the metriplectic dynamics, we have
\bq
\dot{I}= \{I,H\} + (I, S)= (I, S)\,.
\eq
Thus, as pointed out in \cite{pjm86}, this immediately implies that a function that is simultaneously conserved by  both  the full metriplectic
dynamics and its Hamiltonian part, is necessarily conserved by  the dissipative part.
Physically, it may be desirable  for general metriplectic systems to conserve dynamical constraints, i.e., conserved quantitates  of its Hamiltonian part and the examples given in e.g.  \cite{pjm84,pjm84b,pjm86,pjmC20} satisfy this condition  and method  based on multilinear  brackets was given  in  \cite{pjmBR13}.

\subsection{The metriplectic 4-bracket and dynamics}
\label{ssec:mp4bktd}

\subsubsection{The metriplectic 4-bracket}
\label{sssec:curvature4bkt}

To motivate our metriplectic 4-bracket, we begin by supposing our phase space manifold   $\calz$ is a Riemannian manifold with a curvature tensor $R$, 
\bq 
 R \colon\mathfrak{X}(\calz)\times\mathfrak{X}(\calz) \times\mathfrak{X}(\calz)\rightarrow \mathfrak{X}(\calz)\,,
 \label{R1}
 \eq
 i.e, for  $\bfX,\bfY,\bfZ\in \mathfrak{X}(\calz)$,   $R(\bfX,\bfY)\bfZ \in \mathfrak{X}(\calz)$. 
 As usual, we may write the Riemann curvature tensor in coordinate form as

 $\mathbf{R}^i_{\ jk l}$, with $i,j,k,l= 1,2,\dots, \mathrm{dim}\left(\calz\right)$. Given a metric 
 \bq
 g:\mathfrak{X}(\calz)\times\mathfrak{X}(\calz) \rightarrow \La^0(\calz),
 \eq
 which has the usual covariant tensor expression $g_{ij}$, we can construct the totally covariant tensor
 \bq 
 R:\mathfrak{X}(\calz)\times\mathfrak{X}(\calz) \times\mathfrak{X}(\calz)\times \mathfrak{X}(\calz)\rightarrow \La^0(\calz)
  \label{R2}
  \eq
 defined by  
 \bq
 R(\bfX,\bfY,\bfZ,\bfW)=g(\bfR(\bfX,\bfY)\bfZ,\bfW)
 \label{convenR},
 \eq
or in index form $R_{ijkl}= g_{im}R^m_{\ jk l}$. 

The totally covariant Riemann tensor possess the following symmetries :
\bqy
R_{ijkl}&=& -R_{jikl}
\label{asym1}
\\
R_{ijkl}&=&-R_{ijlk} 
\label{asym2}
\\
R_{ijkl}&=&R_{klij}\,.
\label{asym3}
\eqy
\\
(iv) the cyclic or as it is sometimes called the algebraic or first Bianchi identity
\bq
R_{ijkl} + R_{iklj} +R_{iljk} =0
\label{cyclic}
\eq
(v) differential or second Bianchi identity, which has been called the Jacobi  identity
\bq
R_{ijkl\, ;m} + R_{ijlm \, ;k} + R_{ijmk \, ;l} \equiv 0\,.
\eq

{The symbol $R$ in \eqref{R1} and \eqref{R2} is used in different senses. In the remainder of this paper we will use $R$ in other senses as well, ones not even necessarily related to Riemannian curvature. We do this to avoid the proliferation of symbols and  trust the usage will be clear from context. }

Given the above background,  we follow a path analogous to that of Sec.~\ref{ssec:poisson} to motivate our metriplectic bracket. Suppose we are given a fully contravariant tensor 
 \bq
  R:\La^1(\calz)\times\La^1(\calz) \times \La^1(\calz)\times \La^1(\calz)\rightarrow \La^0(\calz)\,,
  \eq
 satisfying the same symmetries as the Riemann tensor. Such an object can be constructed, for example, by composing the fully covariant curvature tensor $R$ with any tangent-cotangent isomorphism, { i.e., raising and lowering operation,} the metric being one obvious choice.  
We have a natural bracket on functions $f,k,g$, and $n$  by 
  \bq
 (f,k;g,n):=R(\bfd f, \bfd k,\bfd g,\bfd n) \,,
 \label{curvatureBkt}
\eq
which is our metriplectic 4-bracket. Restricted to a coordinate neighborhood, the metriplectic 4-bracket can be expressed in index form as
\bq
(f,k;g,n)=R^{ijkl}(z) \frac{\p f}{\p z^i} \frac{\p k}{\p z^j} \frac{\p g}{\p z^k} \frac{\p n}{\p z^l}\,.
\label{4Bktcoords}
\eq
From the above construction leading to \eqref{curvatureBkt} or \eqref{4Bktcoords}, the following algebraic properties are  immediately evident:
\begin{quote}
(i) linearity in all arguments, e.g,
\bq
(f+h,k;g,n)=(h,k;g,n)+(h,k;g,n)
\label{linearity}
\eq
(ii) the algebraic identities/symmetries
\bqy
(f,k;g,n)&=&-(k,f;g,n)
\label{Basym1}
\\
(f,k;g,n)&=&-(f,k;n,g)
\label{Basym2}
\\
(f,k;g,n)&=&(g,n;f,k)
\label{Basym3}
\\
(f,k;g,n) &+& (f,g;n,k)+ (f,n;k,g)=0
\label{Basym4}
\eqy
(iii) derivation in all arguments, e.g., 
\bq
(fh,k;g,n)= f(h,k;g,n)+(f,k;g,n)h
\label{leibniz}
\eq
which is manifest when written in coordinates as in \eqref{4Bktcoords}. Here, as usual, $fh$ denotes pointwise multiplication. 
\end{quote}
Using the above definitions, we can define the contravariant analog of many constructions in standard Riemannian geometry. Of particular dynamical interest is the contravariant sectional curvature defined on 1-forms, say  $\si,\et \in  \La^1(\calz)$, by 
\bq
K(\si, \eta)  \coloneqq  R(\si,\eta, \si,\eta,) = (f,g,f,g) \,,
\label{secCurv}
\eq
where the second equality follows if  $\si=\bfd f$ and $\eta = \bfd g$.  Here we choose to forsake the conventional normalization of $|\si  \wedge \eta|$ for simplicity. Throughout this paper, we will assume that 
\bq
\label{pos-def}
K(\si, \eta) \geq 0
\end{equation}
{In \ref{ssec:KN} we will give a construction that ensures this positive semidefiniteness.

By the above construction, it is clear that a large class of metriplectic 4-brackets exist on Riemannian manifolds. In \ref{sssec:poisson-connect}, we will show that the addition of a Poisson structure leads naturally to such a bracket. For our purposes, it is often convenient to part ways with the underlying geometry requiring only that a metriplectic 4-brackets satisfy the algebraic properties of \eqref{Basym1}, \eqref{Basym2}, and \eqref{Basym3} as well as the linearity and Leibnitz properties. This is equivalent to saying that the metriplectic 4-bracket derives from a 4-tensor satisfying \eqref{asym1}, \eqref{asym2}, and \eqref{asym3}.  {In addition we will assume the positivity  condition of \eqref{pos-def} and }refer to 4-tensors and associated 4-brackets that have these {properties} as being  {\it minimal metriplectic}.   In Sec.\ \ref{ssec:tremoval} we will see that the cyclic identity of \eqref{cyclic} need not be further imposed to be assumed true. The differential Bianchi identity has ramifications akin to those of the  Jacboi identity of Hamiltonian dynamics. These will be elucidated in a future publication. 
 
\subsubsection{Dynamics generated  by the  metriplectic 4-bracket }
\label{ssec:curvature4bkt-metri}

From  the metriplectic 4-bracket of \eqref{curvatureBkt} we construct the symmetric, yet degenerate bracket:
\bq
(f,g)_H \coloneqq (f,H;g,H)=R^{ijkl}\frac{\p f}{\p z^i}\frac{\p H}{\p z^j}\frac{\p g}{\p z^k}\frac{\p H}{\p z^l}\,.
\label{fgH}
\eq 
Interchanging $f$ and $g$ amounts to interchanging $i$ and $k$, and because we have symmetrical contraction in $j$ and $l$, we get by \eqref{Basym3}
\bq
(f,g)_H= (g,f)_H\,.
\eq
Thus  the $G$-metric follows from \eqref{fgH}, viz.
\bq
G^{ik}=R^{ijkl}\frac{\p H}{\p z^j}\frac{\p H}{\p z^l}\,,
\eq
and with this bracket, the dissipative dynamics is generated by
\bq
\dot{z}^i= (z^i,H;S,H)= (z^i,S)_H=G^{ij}\frac{\p S}{\p z^j}\,,
\eq
where, for the full metriplectic dynamics, we would add the Poisson bracket contribution to the above.

Energy conservation comes automatically upon using \eqref{Basym1} and \eqref{Basym2}, i.e., 
\bq
(f,H)_H= (H,f)_H=0\qquad \forall\, f\,.
\eq
Then the entropy dynamics would be governed by 
\bq
\dot{S} =(S,S)_H = (S,H;S,H)\geq 0\,,
\label{entropy}
\eq
where the inequality comes from \eqref{pos-def} and assures entropy production.

Thus we see how  metriplectic 2-brackets  first given  in \cite{pjm84,pjm86} arise from metriplectic 4-brackets.

\subsubsection{Lie-metriplectic manifolds: a  metriplectic 4-bracket view}
\label{sssec:Lie-M}

Since the Poisson manifolds of metriplectic dynamics usually  arise from the standard  picture of  reduction \cite{MWreduction}, we give here some comments on this case. We will refer to manifolds of this type as {\it Lie-metriplectic} manifolds. 

Hamiltonian systems with a configuration  space being a Lie group $G$ can lead to a  reduced phase space $\mathcal{Z}=T^*G/G\cong \mathfrak{g}^*$ where  $\mathfrak{g}^*$ is the  dual of the Lie algebra $\mathfrak{g}$ of  $G$. To endow $G$ with a metriplectic structure, we need a metric $g$ (not to be confused with $g\in \La^0(\calz)$). While many such metrics may exists, left-invariant metrics (metrics constant with respect to the vector field module basis of left-invariant vector fields) are the only ones  that  respect the reduced phase space $\mathcal{Z}$. By left translation, these metrics are globally defined by their action at the identity of $G$. Given such a metric, we may consider the left-invariant curvature tensor $R_G$, which restricts to, and is in fact entirely encoded by, the constant tensor $R_{G}|_e$ on $\mathfrak{g}$.
Metrically raising the indices of $R_G$, the metriplectic 4-bracket on $\mathfrak{g}^*$  takes the form of \eqref{4Bktcoords} with $R^{ijkl} = R_{G}|_e(z^i,z^j,z^k,z^l)$ {and}  $z^i$ being the coordinates of  $\mathfrak{g}^*$. 

{Because cases like the above, where the metriplectic 4-tensor is independent of the coordinate $z$,  have special properties,  we   call these  {\it Lie-metriplectic} 4-brackets.}
For such brackets, the $z$-dependance of the associated metriplectic 2-bracket  is determined by $H$.  For example, when the Hamiltonian is quadratic, say $H = H_{ij}z^i z^j$,  the  metriplectic 2-bracket is a quadratic form.

\subsubsection{Torsion removal -- uniqueness of metriplectic 4-brackets}
\label{ssec:tremoval}

A minimal  metriplectic 4-tensor $A^{ijkl}$ obeying   \eqref{Basym1}, \eqref{Basym2}, and \eqref{Basym3}, but not the cyclic symmetry of \eqref{Basym4} is said to have torsion because this cyclic symmetry can be traced to the symmetry in two of the Christoffel  symbol indices (see Sec.~\ref{sssec:LCforms}) in Riemannian geometry.   
Tensors  that satisfy   \eqref{Basym1}, \eqref{Basym2},  \eqref{Basym3}, and  \eqref{Basym4} are often called algebraic curvature tensors.

Minimal metriplectic tensors  like $A^{ijkl}$ can have their torsion removed  (see e.g.\ \cite{lanczos})  by defining the antisymmetric tensor 
\bq
T^{ijkl} = \frac{1}{3}\left( A^{ijkl}+ A^{iklj}+ A^{iljk}\right) 
\label{Ttensor}
\eq
and using it  to construct 
\bq
R:= A - T\,, 
\label{amt}
\eq
which  does indeed satisfy the algebraic Bianchi identity of  \eqref{Basym4},  and thus is an algebraic curvature tensor.

Because   $T$ is totally antisymmetric,  for any choice of functions $f,g, H$, 
\bq
 R^{ijkl}\frac{\partial f}{\partial z^i}\frac{\partial H}{\partial z^j} \frac{\partial g}{\partial z^k} \frac{\partial H}{\partial z^l}  
 = A^{ijkl}\frac{\partial f}{\partial z^i}\frac{\partial H}{\partial z^j} \frac{\partial g}{\partial z^k} \frac{\partial H}{\partial z^l}\,,
\label{equivalence}
\eq
where $R$  is  given by \eqref{amt}.   Therefore, the metriplectic 2-bracket 
\bq
(f,g)_H=(f,H;g,H)
\eq
 does not see the torsion.  {Although the  metriplectic 2-bracket and the concomitant dynamics it generates do not see torsion, the geometrical structure of the manifold without torsion is decidedly different from that  with torsion. Consequently, the global understanding of the dynamics is facilitated by knowing the torsion can be removed.}

In light  of \eqref{equivalence} we  can use the procedure above  to define a chain of isomorphisms between unique zero torsion metriplectic 4-brackets and minimal   metriplectic tensors obeying \eqref{Basym1}, \eqref{Basym2},  and \eqref{Basym3}.  Suppose we are given a minimal  metriplectic 4-bracket, then, with $T$ according to \eqref{Ttensor}, 
$A+T$  defines the same metriplectic system as $A$. Hence, to define a better space of metriplectic 4-brackets we can use the  quotient (equivalence relation) 
\[
A \sim R \implies A = R + T \,.
\]
In the set of equivalent curvature tensors, there is  a unique  one that is  torsion free for defining a metriplectic 4-bracket.   To  see this  suppose $R$ and $A$ are torsion free algebraic curvature tensors that define the  same metriplectic 2-bracket,  
\[
R^{ijkl}\frac{\p H}{\p z^j}\frac{\p H}{\p z^l} = A^{ijkl}\frac{\p H}{\p z^j} \frac{\p H}{\p z^l} \,,
\]
for every choice of some function $H$. Let  $A=R+T$.  Since algebraic curvature tensors form a vector space, it follows that $T$ is an algebraic curvature tensor satisfying  
\[
T^{ijkl}     \frac{\p H}{\p z^j} \frac{\p H}{\p z^l} = 0 \,,
\]
for all functions $H$. Upon choosing  $H = z^r$,  it follows that 
\bq
T^{irkr} = 0 \,,
\label{rr}
\eq
where $r$ is arbitrary and not summed over. Upon choosing $H = z^r + z^s$  it follows that 
\bq
T^{irkr} +T^{isks} + T^{irks} + T^{iskr} = T^{irks} + T^{iskr} = 0\,,
\label{24asym}
\eq
where the first equality follows by using \eqref{rr} in the first and  second  terms.  
If $T$ did satisfy \eqref{Basym4}, then
\[
T^{ijkl} = -T^{iljk} - T^{iklj} = T^{ilkj} - T^{iklj} \,,
\]
while antisymmetry in the second and fourth slot, because of \eqref{24asym}, would further imply that 
\[
T^{ijkl} =  -T^{ijkl} + T^{ijlk} = -T^{ijkl} - T^{ijkl} \implies T = 0 \,.
\]
Thus {$R=A$} and we see why  the algebraic Bianchi identity of  \eqref{Basym4}  is important and desirable.  
It removes redundancy in the theory. Also, we note it allows $R$ to be written as a sum of the  Kulkarni-Nomizu products described in Sec.\ \ref{ssec:KN}, which we will  see  is a  quite useful tool.

\subsection{Special metriplectic 4-bracket constructions}
\label{ssec:special4bkt}

Consider now some natural 4-bracket constructions. 

\subsubsection{Affine and Levi-Civita forms}
\label{sssec:LCforms}

Given  any  affine manifold we can define the Riemann-Christoffel  curvature tensor
\bq
R^i_{\, jkl}=\Ga^{i}_{\, rk}\Ga^{r}_{\, jl} - \Ga^{i}_{\, rl} \Ga^{r}_{\, jk} 
+ \frac{\p \Ga^{i}_{\, jl}}{\p  z^k} -\frac{\p \Ga^{i}_{\, jk}}{\p z^l} 
\label{RCten}
\eq
and further if our  manifold is Riemannian we have  the usual Levi-Civita connection
\bq
\Ga^{l}_{\, jk}=\frac12 g^{lr} \left(\frac{\p g_{jk}}{\p z^r}  + \frac{\p g_{rj}}{\p z^k}  
 -\frac{\p g_{rk}}{\p z^j} 
\right)\,.
\label{christo}
\eq
Thus,  using the metric $g$ we can construct
\bq
R^{ijkl}=g^{jr}g^{ks}g^{lt}  R^i_{\, rst}
\label{Rcurv}
\eq
and hence obtain a metriplectic 4-bracket  of the form of \eqref{4Bktcoords}.  
This 4-bracket has  the  associate $G$-metric tensor
\bq
G^{ij}=R^{ikjl}\frac{\p H}{\p  z^k}\frac{\p H}{\p  z^l}\,.
\label{RCGtens}
\eq
Using \eqref{RCten}, \eqref{christo}, \eqref{Rcurv} we see that  the $G$-metric of \eqref{RCGtens} is trivially  zero for Euclidean space, but in general it is a complicate  expression in terms of the Riemannan metric $g$, designed to have   $\p H/\p z$ in its  kernel.

As explained in Sec.\ \ref{sssec:curvature4bkt}, this class of 4-brackets motivated  our  theory.  However, our metriplectic construction is based on the algebraic properties of the bracket of  \eqref{curvatureBkt}.  We point out that not all 4-brackets are based on such Riemann curvature tensors.    Below  we give some other constructions of metriplectic 4-brackets.

\subsubsection{Kulkarni-Nomizu construction}
\label{ssec:KN}

Curvature 4-brackets with the requisite symmetries can be easily constructed by making use of 
the Kulkarni-Nomizu (K-N) product  \cite{kulkarni,nomizu} (anticipated in \cite{lanczos}).   Consistent with the bracket formulation of Sec.~\ref{sssec:curvature4bkt} we deviate from convention for K-N products and work on the dual space.  Given two symmetric bivector fields, say $\sigma$ and $ \mu $, operating on 1-forms $\bfd f, \bfd k$ and $\bfd g,\bfd n$, the K-N product is defined by 
\bqy
 \sigma \KN \mu \,(\bfd f, \bfd k,\bfd g,\bfd n)&=&  \sigma(\bfd f,\bfd g) \, \mu (\bfd k,\bfd n) 
\nonumber\\
&-& \sigma(\bfd f,\bfd n) \, \mu (\bfd k,\bfd g)
 \nonumber\\
&+&  \mu (\bfd f,\bfd g)\,  \sigma(\bfd k,\bfd n)
 \nonumber \\
&-&  \mu (\bfd f,\bfd n)\,  \sigma(\bfd k,\bfd g) \,.
\label{KNfinite}
\eqy
Thus, we  may define a 4-bracket according to 
\bq
(f,k;g,n)= \sigma \KN  \mu (\bfd f, \bfd k,\bfd g,\bfd n)\,.
\eq
In coordinates this gives a 4-bracket of the form \eqref{4Bktcoords} with 
\bq
R^{ijkl}=\sigma^{ik}  \mu  ^{jl} -  \sigma^{il}  \mu ^{jk} +  \mu ^{ik} \sigma^{jl}  -  \mu ^{il} \sigma^{jk}\,.
\label{KN4bkt}
\eq
It is easy to show that such a bracket has all  of the algebraic  symmetries described  in Sec.\ \ref{sssec:curvature4bkt}.
{{In addition}, it can be shown using the Cauchy-Schwarz inequality that positivity of the sectional curvature  is satisfied,   if both $\si$ and $\mu$ are positive semidefinite. {Moreover, if one of $\si$ or $\mu$ is positive definite, thus defining an inner product, then the sectional curvature of \eqref{pos-def} satisfies $K(\si, \eta) \geq 0$ with equality if and only if  $\si \propto \eta$.} (See  \citep{fiedler03} for additional results along these lines,  including theorems about completeness of K-N types of bases.)
}  Thus, it is easy to build minimal metriplectic 4-brackets.

If one chooses both  $\sigma$ and $ \mu$  to  be proportional to the metric of a Riemannian manifold, then \eqref{KN4bkt}  reduces  to 
\bq
R^{ijkl}=K\left( g^{ik}g^{jl}- g^{il}g^{jk} \right) \,,
\label{gProj}
\eq 
which is  the curvature associated with a form of  metriplectic bracket first given in \cite{pjm09a} (cf.\  Eq.\ (38) {of that reference}).  In the case  were $g$ is Euclidean  this yields the metriplectic  4-bracket
\bq
(f,k;g,n) =K \left( \de^{ik}\de^{jl}- \de^{il}\de^{jk} \right)   \frac{\p f}{\p z^i}  \frac{\p k}{\p z^j}
\frac{\p g}{\p z^k}  \frac{\p n}{\p z^l}\,,
\label{so3}
\eq
whence $(f,H;g,H)$ produces the metriplectic bracket for the rigid body given in \cite{pjm86}.   A Riemannian manifold is called a {\it space form} \cite{goldberg}   if its sectional curvature is equal to a constant, say $K$. The above  are thus space form metriplectic  4-brackets. 

In the cases of \eqref{gProj} and \eqref{so3}, the metric tensor $G^{ij}$ of \eqref{metriplectic_dynamics} is simply the projector that  projects out $\p H/\p z^i$ using $g^{ij}$  and $\de^{ij}$, respectively.

\subsubsection{Lie algebra based metriplectic 4-brackets}
\label{sssec:LAMbkts}

Metriplectic brackets associated with Lie algebras were first investigated in \cite{pjm09a} and later in \cite{pjmBR13}.  Here we give two natural Lie algebra related constructions for metriplectic 4-brackets. 

First, given  any Lie algebra $\mathfrak{g}$ with structure constants $c^{ij}_{\, \ k}$ and  a symmetric semidefinite tensor $g^{rs}$ one can construct a 4-bracket based on these quantities as follows:
\bq
(f,k;g,n) = c^{ij}_{\, \ r} c^{kl}_{\, \ s} \, g^{rs} \frac{\p f}{\p z^i}\frac{\p k}{\p z^j}\frac{\p g}{\p z^k}\frac{\p n}{\p z^l}\,.
\label{LA4bkt}
\eq
It is easy to see that this  bracket is  minimal  metriplectic, obeying  \eqref{Basym1}, \eqref{Basym2}, and \eqref{Basym3}, but  it does not have the cyclic symmetry of \eqref{Basym4}.   However,  this symmetry can be  obtained, i.e.,  the  torsion removed,  by the procedure of Sec.\ \ref{ssec:tremoval}, where 
$A^{ijkl}=c^{ij}_{\, \ r} c^{kl}_{\, \ s} \, g^{rs}$. Using the  notation  of Sec.\ \ref{ssec:tremoval}, we have
\[
T^{ijkl} = \frac{1}{3} g^{rs} (c^{ik}_{\, \ r} c^{lj}_{\, \ s} + c^{il}_{\, \ r} c^{jk}_{\, \ s}  + c^{ij}_{\, \ r} c^{kl}_{\, \ s} ).
\]
Thus it follows that $A$ is equivalent to the following algebraic curvature tensor
\bq
 B^{ijkl} = \frac{g^{rs} }{3} (2 c^{ij}_{\, \ r} c^{kl}_{\, \ s}+c^{ik}_{\, \ r} c^{jl}_{\, \ s} - c^{il}_{\, \ r} c^{jk}_{\, \ s} ). 
\label{Btensor}
\eq
Furthermore, $A$ is also equivalent to the following minimal  metriplectic tensors:
\bal
A^{ijkl} &\sim \frac{g^{rs} }{2} (c^{ij}_{\, \ r} c^{kl}_{\, \ s}-c^{ik}_{\, \ r} c^{lj}_{\, \ s} - c^{il}_{\, \ r} c^{jk}_{\, \ s} )
\ncr
 &\sim -g^{rs} \left( c^{ik}_{\, \ r} c^{lj}_{\, \ s} + c^{il}_{\, \ r} c^{jk}_{\, \ s}  \right)
\eal
and so on.   

In this  construction care must be  taken in ensuring that the null space and signature of $g^{rs}$, so far only assumed symmetric,  does not lead to  undesirable effects, such as preventing the desired relaxation to equilibrium. The Euclidean metric $g^{rs}=\de^{rs}$ is the simplest choice that alleviates these problems, but any  metric  is a possibility.  

As a second case, a refinement of the first, suppose the tensor $g^{rs}$  is proportional to the  Cartan-Killing form, i.e.,  $g_{CK}^{rs}=\la c^{rm}_{\, \ n} c^{sn}_{\, \ m}$ for constant $\la$, as considered in \cite{pjm09a}.  Recall, for semisimple Lie algebras $g_{CK}$  has no kernel and thus {it possesses an inverse,}  and for compact  semisimple Lie algebras like $\mathfrak{so}(3)$ it  is  in addition definite. With the choice of $g_{CK}$ the bracket of \eqref{LA4bkt} is naturally associated to any Lie algebra with no additional structure needed, akin  to  the bracket given in \cite{pjm09a}.  

For  the Lie algebra $\mathfrak{so}(3)$, $g_{CK}^{rs}\sim \de^{rs}$,  and the bracket using  \eqref{Btensor} reduces to the rigid body bracket of \eqref{so3}. In general, we find the following upon inserting $g_{CK}$ into \eqref{LA4bkt}:
\bq
(f,k;g,n) = \la\,  c^{ij}_{\, \ r} c^{kl}_{\, \ s} \, c^{rm}_{\, \ n} c^{sn}_{\, \ m}\, \frac{\p f}{\p z^i}\frac{\p k}{\p z^j}\frac{\p g}{\p z^k}\frac{\p n}{\p z^l}\,.
\label{LA4bktCK}
\eq
Using the Jacobi identity on the terms $c^{ij}_{\, \ r} \, c^{rm}_{\, \ n}$ and $ c^{kl}_{\, \ s} \,  c^{sn}_{\, \ m}$, gives an expression that can be manipulated into a tensor of the form
\bq
 B_{CK}^{ijkl} = g^{rs}_{CK} \big(2 c^{ij}_{\, \ r} c^{kl}_{\, \ s}-c^{ik}_{\, \ r} c^{lj}_{\, \ s} - c^{il}_{\, \ r} c^{jk}_{\, \ s} \big)\,. 
\label{BtensorCK}
\eq
Thus, for this case, torsion is already removed.  Moreover, it can further be shown to take the form of \eqref{LA4bkt} with metric $g_{CK}$, as was the case for $\mathfrak{so}(3)$ resulting in \eqref{LA4bktCK}. This follows from the fact that $c^{ijk}$ is antisymmetric under any interchange of indices, i.e., 
\bq
g_{CK}^{rs} c_{\, \ s}^{ij} = c^{ijr} = -c^{irj} = -g_{CK}^{js}c^{ir}_{\, \ s}\,.
\label{CKid}
\eq
Using this and the Jacobi identity 
\[
c^{is}_{r} c^{jk}_s = c^{js}_{r} c^{ik}_s -c^{ks}_{r} c^{ij}_s
\]
the middle term of \eqref{Btensor} satisfies
\bq
g_{CK}^{rs} c^{ik}_r  c^{jl}_s= g_{CK}^{rs} c^{il}_r c^{jk}_{s}  + c^{ij}_r c^{kl}_{s}
\eq
and  consequently 
\bq
B^{ijkl} =g_{CK}^{rs}c^{ij}_r c^{kl}_s \,.
\eq
{Recall, we referred to  4-brackets defined by such $z$-independent 4-tensors  as    { Lie-metriplectic}.}

\subsubsection{On metriplectic geometry}
\label{sssec:poisson-connect}

Given the from of metriplectic dynamics of \eqref{metriplectic_dynamics} it is natural to explore  manifolds with both Poissonian and Riemannian structure.  Such manifolds are plentiful because a metric exists on any Poisson manifold (assuming Hausdorff, paracompactness). As a guiding principle of Poisson geometry, one often looks to generalize objects defined on the tangent bundle to the cotangent bundle. Following this principle, we introduce the cotangent analogs of connections and curvature. For the sake of brevity, we omit much of the motivation, mathematical detail, and theoretical importance of such objects. Instead, we introduce the ideas of contravariant connections and contravariant curvature {axiomatically \cite{koszul}} and refer the interested reader to  \cite{holonomy}, which served as a main motivation for us.  (See also \cite{alioune}.)

Let $\mathcal{Z}$ be a manifold with a Poisson bracket 
$\{\cdot, \cdot \}\colon C^\infty(\mathcal{Z}) \times  C^\infty(\mathcal{Z})  \to  C^\infty(\mathcal{Z})$. We will work locally on a coordinate patch, denoting the coordinate functions $z^i$ as  above. In order to define a curvature tensor on forms, it is desirable to extend the Poisson bracket to a Lie bracket on forms  $[\cdot,\cdot]^{J}\colon \La^1(\mathcal{Z}) \times \La^1(\mathcal{Z}) \to \La^1(\mathcal{Z})$,  which is done according to  the formula
\bq
[\mathbf{d}f, \mathbf{d} g]^J := \mathbf{d}\{f,g \} {=\mathbf{d}\big(J(\bfd f, \bfd g)\big)}  \,,
\eq
where we use the superscript $J$ to distinguish this bracket and $f,g$ are functions, i.e., 0-forms,  (see \cite{holonomy}  for an extension beyond exact forms).
Furthermore, the Poisson bracket provides the natural Poisson tensor  {of Sec.\ \ref{ssec:poisson}} defined by 
$J^{ij} = \{z^i,z^j \}$,  which  is useful in that it allows us to define a map 
$J \colon \La^1(\mathcal{Z}) \to \mathfrak{X}(\mathcal{Z})$ 
from   1-forms to vector fields by
\bq
(J \al)^j \coloneqq \alpha_i J^{ij}\,,
\eq
where $\alpha$ is a 1-form.

In conventional Riemannian geometry, one defines a covariant connection as a map between vector fields, $\nabla\colon \mathfrak{X} \times \mathfrak{X} \to \mathfrak{X}$,  satisfying some linearity properties. However, given the singular submanifold structures inherent to Poisson geometry, it is sometimes a matter of mathematical necessity to extend the notion of a connection to the cotangent bundle.  There, a contravariant connection is a map 
$D\colon  \La^1(\mathcal{Z}) \times \La^1(\mathcal{Z}) \to \La^1(\mathcal{Z}) $
 satisfying the same identities as the covariant connection. Namely, given $\alpha, \beta, \gamma \in \La^1(\mathcal{Z})$ and  $f\in \La^0(\mathcal{Z})$
\begin{align}
D_{\alpha + \beta} \gamma &= D_{\alpha} \gamma + D_{\beta} \gamma   \\ 
D_{f\alpha} \gamma &= fD_\alpha \gamma   \\
D_\alpha (\beta + \gamma) &= D_\alpha\beta + D_{\alpha} \gamma   \\
D_{\alpha}(f\gamma) &= f D_{\alpha}\gamma + J(\alpha)[f]\gamma
\label{KLeib}\,.   
\end{align}
In \eqref{KLeib},  {$J(\alpha)[f]=\alpha_iJ^{ij}\p f/\p z^j$ is a 0-form  that  replaces the  term  ${\bf{X}}(f)$ in  Koszul's  algebraic  Leibniz identity.

 Given a choice of contravariant connection, it is natural to define the contravariant curvature 
$R\colon  \La^1(\mathcal{Z}) \times \La^1(\mathcal{Z}) \times \La^1(\mathcal{Z}) \to \La^1(\mathcal{Z}) $ by  obvious analogy to standard Riemannian geometry
\bqy
R(\alpha, \beta) \gamma &=& D_\alpha D_\beta \gamma - D_\beta D_\alpha \gamma - D_{[\alpha, \beta]^J}\gamma
\nonumber\\
&=& (R(\alpha, \beta) \gamma)_l \, \mathbf{d}z^l \, .
\label{riemannJ}
\eqy
To get the coordinate form of the contravariant curvature we simply define
\bq
R^{ijk}_{\; \; \; \; \; l} := \left( R(\mathbf{d}z^i, \mathbf{d}z^i) \mathbf{d}z^k\right)_l\,.
\label{RJgtensor}
\eq
Of particular interest to us are metric connections. 

Given a metric $g$, there is a  Levi-Civita-like contravariant connection given by the formula 
\bqy
2g(D_{\alpha} \beta, \gamma) &=& J(\alpha)\big[g(\beta, \gamma)\big] - J(\gamma)\big[g(\alpha, \beta)\big] 
 \nonumber \\ 
 && + J(\beta)\big[g(\gamma, \alpha)\big] + g\big([\alpha, \beta]^J, \gamma\big) 
 \nonumber\\
 &&
 - g\big([\beta, \gamma]^J, \alpha\big) + g\big([\gamma, \alpha]^J, \beta\big)\,, 
\eqy
which  has the  coordinate form,
\bqy
&&2g^{r s}(D_{\mathbf{d}f} \mathbf{d}h)_s  =
 \\
&&\quad  \qquad
 \frac{\partial f}{\partial z^i} J^{ij}\frac{\partial}{\partial z^j}\left[g^{kr} \frac{\partial h}{\partial z^k} \right] + J^{i r }\frac{\partial}{\partial z^i}\left[g^{kl} \frac{\partial f}{\partial z^k} \frac{\partial h}{\partial z^l}\right] \nonumber \\  
&& \qquad 
+  \frac{\partial h}{\partial z^i} J^{ij}\frac{\partial}{\partial z^j}\left[g^{kr }  \frac{\partial f}{\partial z^k}\right]+ g^{kr}\frac{\partial}{\partial z^k}\left[ J^{ij}\frac{\partial f}{\partial z^i}\frac{\partial g}{\partial z^j}\right]  \nonumber  \\
&& \qquad 
 -  g^{kl}\frac{\partial}{\partial z^k}\left[ J^{ir}\frac{\partial h}{\partial z^i}\right]   \frac{\partial f}{\partial z^l} -  g^{kl}\frac{\partial}{\partial z^k}\left[ J^{jk }\frac{\partial f}{\partial z^j}\right]   \frac{\partial h}{\partial z^l}. \nonumber
\eqy
This formula is perhaps best understood as defining the contravariant Christoffel symbols 
\bqy
\Gamma^{ij}_{\, \ l} &\coloneqq&  (D_{\mathbf{d}z^i} \mathbf{d}z^j)_l
\nonumber\\
&=& \frac{1}{2} g_{kl}\left[ J^{is}\frac{\partial g^{jk}}{\partial z^s}- J^{ k s}\frac{\partial g^{ij}}{\partial z^s} + J^{js}\frac{\partial g^{ik}}{\partial z^s}\right] \nonumber \\ && +\frac{1}{2} g_{kl}\left[  g^{ks}\frac{\partial J^{ij}}{\partial z^s}   -  g^{si}\frac{\partial J^{jk}}{\partial z^s}  -  g^{sj}\frac{\partial J^{ik}}{\partial z^s}  \right].
\label{GaJ}
\eqy

Just like its covariant analog, the contravariant Levi-Civita connection is the unique connection that is both torsion-free, 
\begin{align}
D_\alpha \beta - D_\beta \alpha = [\alpha, \beta]^J \  \longleftrightarrow
\  \Gamma^{ij}_{\, \ k} - \Gamma^{ji}_{\, \ k} = \frac{\partial J^{ij}}{\partial z^k} \,,
\label{tfree}
\end{align}
and metric compatible  {(vanishing covariant derivative of the metric)}
\bqy
&& J (\alpha) [g(\beta, \gamma) ] = g(D_\alpha \beta, \gamma) + g(\beta, D_\alpha \gamma)
\nonumber\\
&&\qquad \longleftrightarrow \ J^{is}\frac{\partial g^{jk}}{\partial z^j}   = g^{ks}\Gamma^{ij}_{\, \ s}   + g^{js} \Gamma_{\, \ s}^{ik}\,.
 \label{compat}
\eqy
Explicitly, for  1-forms $\alpha = \alpha_i dz^i$ and $\beta=\beta_j dz^j$ we compute 
 \bal
 D_{\alpha}\beta &=\alpha_i D_{dz^i} [\beta_j dz^j]
 \label{covcon}
 \\
 &= \alpha_i \beta_j D_{dz^i} [ dz^j] + \alpha_i J(dz^i)[\beta_j] dz^j 
 \ncr
 &= dz^j \alpha_i J^{is}\frac{\partial \be_j}{\p z^s} + \alpha_i \beta_j \Gamma^{ij}_k dz^k\,.
 \nonumber
 \eal
 Thus, \eqref{riemannJ}, produces the tensor
 \bal R^{ijk}_{\; \; \; l}&=  \Gamma^{jk}_s \Gamma^{is}_{\, \ l}   - \Gamma^{ik}_{\, \ s} \Gamma^{js}_{\, \ l} 
 \ncr
 &\qquad  - \frac{\partial J^{ij}}{\p z^s} \Gamma^{sk}_{\,  \ l}   + J^{is} \frac{\partial \Gamma^{jk}_{\, \ l}}{\p z^s} -J^{js}\frac{\partial \Gamma^{ik}_{\, \ l}}{\p z^s}\,.
\label{curvJg}
 \eal
The addition of a metric structure allows us to raise the indices and obtain the fully  contravariant 4-tensor from \eqref{curvJg} according to
\bq
R^{ijkl}\coloneqq R^{ijk}_{\; \; \; \; \; s}\,  g^{sl}\,.
\eq
Provided we use the  contravariant connection described  above in \eqref{riemannJ}, $R^{ijkl}$ obeys all the symmetries of the normal Riemann tensor, including the first and second Bianchi identities. Further, by raising an index, we have a map $R: \La^1(\mathcal{Z}) \times \La^1(\mathcal{Z}) \times \La^1(\mathcal{Z}) \times \La^1(\mathcal{Z}) \to \La^0(\calz)$ from forms to scalars. This induces a  4-bracket on functions $( \cdot,\cdot\,; \, \cdot, \cdot):  C^\infty(\mathcal{Z}) \times C^\infty(\mathcal{Z}) \times C^\infty(\mathcal{Z}) \times C^\infty(\mathcal{Z}) \rightarrow C^\infty(\mathcal{Z})$ under the association of a function with its differential
\bqy
(f,g; k, n)  &:=&  R(\mathbf{d}f,\mathbf{d}g,\mathbf{d}k, \mathbf{d}n)
\nonumber\\
&=&   R^{ijkl}\frac{\partial f}{\partial z^i}\frac{\partial k}{\partial z^j}\frac{\partial g}{\partial z^k}\frac{\partial n}{\partial z^l} \,.\nonumber
\eqy

We note that the contravariant connections can behave quite differently from what  one might expect in  Riemannian geometry. For example, consider the Poisson manifold $\mathcal{Z} = \mathfrak{so}(3)$ with the standard Poisson bracket $\{z^i, z^j\} = -\epsilon^{ijk}z^k$. Even with a seemingly flat metric $g^{ij} = \delta^{ij}$, we have the nontrivial Christoffel symbol and nontrivial curvature tensors
\bq
\Gamma^{ij}_{\, \ k} = -\frac{1}{2}\epsilon^{ijk}
\quad\mathrm{and}\quad
R^{ijkl} = \frac{1}{4}(\delta^{ik}\delta^{jl}-\delta^{il}\delta^{jk} )\,.
\label{nontriv}
\eq
Thus, we see again  the  emergence of the metriplectic 4-bracket whose associated 2-bracket was given  \cite{pjm86} and later used in \cite{pjmM18} to model a controlling torque for the free rigid body.

As an aside, we note the structure outlined above can be a powerful tool in the study of Poisson manifolds. For example, if the metric and Poisson structure are compatible  {(vanishing covariant derivative of $J$), i.e.\}  for all $\alpha, \beta , \gamma \in \La^1(\mathcal{Z})$
\bq
J(D_{\alpha}\beta , \gamma) + J(\beta, D_{\alpha} \gamma) = 0\,,
\label{Pcompat}
\eq
then the symplectic leaves of $\mathcal{Z}$ become K\"ahler manifolds (see, e.g.,  \cite{boucetta03}). 

While the metriplectic dynamics generated by compatible Poisson and Riemannian structures promises to be very theoretically interesting, such a condition is  too strong for our current purposes. For example, when $J$ is Lie-Poisson it is easy to verify that the corresponding Cartan-Killing metric is never compatible with $J$. 

Special cases of the metriplectic manifolds of this section come to mind:   Lie-Poisson manifolds with an unidentified metric tensor, Poisson manifolds with a constant metric tensor,  Lie-Poisson manifolds with a  constant Euclidean metric tensor or  the Cartan-Killing metric,   $g_{CK}$.  

The first case follows upon inserting the Lie-Poisson form for $J$ into \eqref{GaJ} and the result into \eqref{curvJg}.  This yields an interesting expression that we will not record here.  The connection for case of constant metric follows immediately from \eqref{GaJ},
viz.
\bq
\Gamma^{ij}_{\, \ l} = \frac{1}{2} g_{kl}\left[  g^{ks}\frac{\partial J^{ij}}{\partial z^s}   -  g^{si}\frac{\partial J^{jk}}{\partial z^s}  -  g^{sj}\frac{\partial J^{ik}}{\partial z^s}  \right].
\label{conGaJ}
\eq
If $J$ is Lie-Poisson this becomes 
\bq
\Gamma^{ij}_{\, \ l} = \frac{1}{2} g_{kl}\left[  g^{ks} c^{ij}_{\, \ s}    -  g^{si} c^{jk}_{\, \ s}  -  g^{sj}  c^{ik}_{\, \ s}  \right].
\label{conGaJLP}
\eq
and if the Euclidean metric $g^{rs}=\de^{rs}$ is inserted into \eqref{conGaJLP}, the following simplified connection is obtained:
\bq
 \Gamma^{ij}_{\, \ k} = \frac{1}{2}(c^{ij}_{\, \ k} - c^{jk}_{\, \ i} + c^{ki}_{\, \ j})\,,
\label{eucGALP}
\eq
(cf.\   \cite{milnor} where  a similar formula is found).  Note  here and henceforth index placement purity  is returned  by inserting appropriate factors  of the metric.  The   curvature tensor  following from \eqref{eucGALP} is
\bqy
R_{\; \; \; \; \; b}^{kij} &=& R_b (dx^i, dx^j)dx^k 
\\
&=& 
\frac{1}{4}(c^{jk}_a - c^{ka}_j + c^{aj}_k)(c^{ia}_b -  c^{ab}_i + c^{bi}_a) 
\nonumber\\
&&\qquad - \frac{1}{4}(c^{ik}_a - c^{ka}_i + c^{ai}_k)(c^{ja}_b - c^{ab}_j + c^{bj}_a) 
\nonumber\\
&&\qquad - \frac{1}{2}c^{ij}_a(c^{ak}_b - c^{kb}_a + c^{ba}_k)\,,
\nonumber
\eqy
where we would raise $b$ with $\de^{bl}$ to obtain the 4-tensor for the corresponding  metriplectic 4-bracket. Finally, if 
the  $g_{CK}$ metric, as  discussed in Sec.\ \ref{sssec:LAMbkts} is assumed, then a simple form for  the 4-tensor is obtained, 
\bq
R^{lkij} = \frac{1}{4}c^{jk}_ac^{ial} - \frac{1}{4}c^{ik}_ac^{jal} + \frac{1}{2}c^{ij}_ac^{kal}
\label{CKR}
\eq
which as  we have  noted reduces to 
\[
 R^{ijkl}= \frac{3}{4} c^{ij}_a \, c^{kl}_a.
 \]

For later use in Sec.\ \ref{sssec:double} we record some useful lemmas about Casimirs.
If $S$ is a Casimir and  {$f,g\in\La_0$ are arbitrary functions,} then 
\bal
D_{\bfd S} \bfd  f  &= D_{\bfd  S} \left[\frac{\p f}{\p z^j} \bfd z^j\right] 
\ncr
&= \bfd z^j J(\bfd S)\left[\frac{\p f}{\p z^j} \bfd z^j\right]  + \left[\frac{\p f}{\p z^j} \bfd z^j\right]  D_{\bfd S} \bfd z^j 
\ncr
&= \frac{\p S}{\p z^i} \frac{\p f}{\p z^j} 
D_{\bfd z^i} \bfd z^j = \frac{\p S}{\p z^i} \frac{\p f}{\p z^j} \Gamma^{ij}_{\, \ k} \bfd z^k. 
\eal
Thus, we have the symmetry 
\bq
D_{\bfd S}\bfd f = \bfd\{S,f\} + D_{\bfd f}\bfd S = D_{\bfd f} \bfd S\,,
\eq
where  $\{S,f\}=0$ because $S$ is a Casimir.   
 Furthermore, we have the antisymmetric bracket
\bal
g(D_{\bfd S} \bfd f, \bfd g) &= - g(\bfd f, D_{\bfd S} \bfd g) +J(\bfd S)[g(\bfd f,\bfd g)] 
\ncr
&= - g(\bfd f, D_{\bfd S} \bfd g)
\eal
and 
\bal
R(\bfd S,\bfd f)\bfd g &= D_{\bfd S} D_{\bfd f} \bfd g - D_{\bfd f} D_{\bfd S} \bfd g - D_{\bfd\{S,f\}} \bfd g 
\ncr
&= D_{\bfd S} D_{\bfd f} \bfd g - D_{\bfd f} D_{\bfd  S} \bfd g \,,
\eal
These properties make Casimirs very special function with respect to the contravariant Riemann tensor and hence the metriplecic 4-bracket.

\subsection{Relation to other dissipation bracket formalisms}
\label{ssec:reductionsF}

The metriplectic 4-bracket  provides a unifying picture, tying together other brackets for dissipation that have previously  appeared in the literature.

\subsubsection{Reduction to Kaufman-Morrison dynamics}
\label{sssec:KM}

In \cite{pjmK82} a bracket for describing the so-called quasilinear theory of plasma physics, a dissipative relaxation theory, was proposed.  The bracket of this theory, here called the KM bracket,  is bilinear and antisymmteric, consequently degenerate, with dynamics generated by a Hamiltonian, $H$.   We denote the bracket of this theory by $[\,\cdot\,  ,\, \cdot\, ]_S: \La^0(\calz)\times \La^0(\calz)\rightarrow  \La^0(\calz)$, where the subscript $S$ will become clear momentarily.   The KM bracket generates dynamics according to 
\bq
\dot{z}^i=[z^i,H]_S\,,
\eq
with the properties that $\dot{H}=0$ and $\dot{S}\geq 0$.  Energy conservation  follows from the antisymmetry of the KM bracket, while the entropy production was built into the theory. 

It is apparent that the KM bracket emerges naturally from any metriplectic 4-bracket  as follows:
\bal
[f,g]_S\coloneqq (f,g;S,H)&=\frac{\p f}{\p z^i}\frac{\p g}{\p z^j}\frac{\p S}{\p z^k}\frac{\p H}{\p z^l}R^{ijkl}
\ncr
&= J^{ij}_{KM} \frac{\p f}{\p z^i}\frac{\p g}{\p z^j}\,,
\eal
where $J^{ij}_{KM}$ is the antisymmetric bivector is given by 
\bq
J^{ij}_{KM} = \frac{\p S}{\p z^k}\frac{\p H}{\p z^l}R^{ijkl}\,.
\eq
Thus, clearly, we have the following:
\bq
[f,g]_S=-[g,f]_S
\eq
by \eqref{Basym1} or \eqref{asym1}, consequently, 
\bq
\dot{H}= [H,H]_S= (H,H;S,H)=0\,,
\eq
and 
\bq
\dot{S}=[S,H]_S=(S,H;S,H)\geq 0\,,
\eq
 by \eqref{entropy}, as was proposed in \cite{pjmK82}.

\subsubsection{Reduction to double bracket dynamics}
\label{sssec:double}

 Double brackets were proposed in \cite{vallis,brockett}  as a computational means of relaxing to equilibria by extremizing a Hamiltonian at fixed Casimirs by using the square of the Poisson tensor $J$ to generate dynamics. The formalism was improved in \cite{pjmF11} and subsequently used in a variety of magnetohydrodynamics contexts in \cite{pjmF17,pjmFWI18,pjmF22}. With the improvements given in \cite{pjmF11}  we can write this dynamics as follows:
 \bq
 \dot{z}^i= ((z^i,H))
 \eq
 where the double  bracket $(( \, .\,,\, .\,))\colon C^\infty(\mathcal{Z}) \times  C^\infty(\mathcal{Z})  \to  C^\infty(\mathcal{Z})$  has the coordinate represention 
 \bq
 ((f,g))= J^{ik} g_{kl}J^{jl}  \frac{\p f}{\p z^i}\frac{\p g}{\p z^j}\,.
 \label{DBg}
 \eq
For metric $g$, evidently
\bq
\dot{H}=((H,H)) \geq 0 \andq \dot{C}=0\,, 
\label{dbprops}
\eq
where $C$ is any Casimir of $J$. A commonly used case of  \eqref{DBg} is that  for Lie-Poisson systems,  where it takes the form
 \bq
 ((f,g))= c^{ik}_{\, \ r} c^{jl} _{\, \ s} \, g_{kl}\, z^s z^r \frac{\p f}{\p z^i}\frac{\p g}{\p z^j}\,.
 \label{LPDBg}
 \eq

One connection between double brackets and the metriplectic 4-bracket can be  made by simply interchanging the role of $H$ with a Casimir $S$, i.e., consider dynamics generated  by the symmetric bracket
\bq
(f,g)_S=(f,S;g,S)\,,
\eq
which for a Casimir $S$ will satisfy  the conditions of \eqref{dbprops}.   However, if $C$ is another Casimir, distinct from $S$, then there is no guarantee it is conserved. We will see in a moment  that the development of Sec.\ \ref{sssec:poisson-connect} provides a way to improved upon this. 

 A direct relationship between  Lie-Poisson double brackets of the form of \eqref{LPDBg} and  metriplectic 4-brackets  of the form of \eqref{LA4bkt} can be made for Cartan-Killing metrics by choosing the entropy
 \bq
 S_{LP}= \frac12 z^a \bar{g}_{ab} z^b
 \eq
 and inserting this into \eqref{LA4bkt}, yielding
 \bal
 (f,g)_{S_{LP}}&=(f,S_{LP};g,S_{LP})
 \ncr
 &=  c^{ij}_{\, \ r} c^{kl}_{\, \ s} \, g^{rs} \frac{\p f}{\p z^i}\frac{\p S_{LP}}{\p z^j}\frac{\p g}{\p z^k}\frac{\p S_{LP}}{\p z^l}
 \ncr
 &=c^{ij}_{\, \ r} c^{kl}_{\, \ s} \, g^{rs}\, \bar{g}_{ja} z^a\, \bar{g}_{lb} z^b\,  \frac{\p f}{\p z^i}\frac{\p g}{\p z^k} \,.
 \label{preJgJ}
 \eal
 Now if we suppose $g$ is the Cartan-Killing metric and $\bar{g}$ is its assumed inverse, \eqref{preJgJ}  becomes upon using   \eqref{CKid}
 \bal
   (f,g)_{S_{LP}}&=  c^{ij}_{\, \ r} c^{rk}_{\, \ s} \, g^{ls}\,  {g}_{ja} z^a\,  {g}_{lb} z^b\,  \frac{\p f}{\p z^i}\frac{\p g}{\p z^k} \,.
  \ncr
    &=  c^{ij}_{\, \ r}  c^{rk}_{\, \ s} \, g_{ja}\,   z^a   z^s\,  \frac{\p f}{\p z^i}\frac{\p g}{\p z^k} 
  \ncr
   &=  c^{ij}_{\, \ r}  J^{rk} \, g_{ja}\,   z^a      \frac{\p f}{\p z^i}\frac{\p g}{\p z^k} 
   \ncr
  &= J^{ir}g_{rs} J^{ks}  \frac{\p f}{\p z^i}\frac{\p g}{\p z^k} \,.
 \eal

Now let us see what transpires when we use the identities at the end of Sec.\ \ref{sssec:poisson-connect}.  If $S$ and $C$ are any Casimirs and our manifold has the metric and Poisson bracket as described, then we can compute as follows:
\bal
(S,C; S,C) &= g(dC, D_{dS} D_{dC}dS - D_{dC} D_{dS}dS)
\ncr
& = -g( D_{dS} dC, D_{dC} dS) + g(D_{dC}dC, D_{dS}dS) 
\ncr
&= -g( D_{dC} dS, D_{dC} dS) + g(D_{dC}dC, D_{dS}dS) 
\ncr
&= \left(\frac{\p C}{\p z^i} \frac{\p C}{\p z^j}\frac{\p S}{\p z^k} \frac{\p  S}{\p z^l} 
- \frac{\p C}{\p z^i} \frac{\p C}{\p z^k} \frac{\p S}{\p z^j} \frac{\p S}{\p z^l}\right) 
\ncr
&\hspace{3cm} \times \Gamma^{ij}_{\, \ a} \, g^{ab}\, \Gamma_{\, \ b}^{kl}\,, 
\eal
a perspicuous form. 

We now specialize to  Lie-Poisson systems with such  constant Cartan-Killing metrics, which gives a  prevalent case  {that we termed}  Lie-metriplect. 
For any Lie-metriplectic system, the Christoffel symbol takes the form of \eqref{conGaJLP}. 
Trivially, now, we see that when both $S$ and $C$ are Casimirs we are left with 
\bq
(S,C; S,C)= 0 \,.
\eq
Thus for { these}  Lie-metriplectic systems, double brackets emerge nicely from our metriplectic 4-bracket. 

\subsubsection{GENERIC is  Metriplectic}
\label{ssec:GM}

In this section we place the ideas given in  \cite{grm84} for a bracket for  the  Boltzmann  collision operator  into a finite-dimensional setting.  Since this work was the origin of a trail leading to what was later referred to As GENERIC, we call this bracket the GENERIC bracket.  We will show given assumptions how to transform the GENERIC bracket, which is not symmetric and not bilinear, into a metriplectic 2-bracket that has these properties. 

Apparently, the latest rendition of {GENERIC  \cite{grm18}} is  written in terms of a dissipation  potential  $\Xi(z,z_*)$, where the shorthand $z_{*i}= \p S/\p z^i$, for some entropy function $S$,    is used.  Using $\Xi(z,z_*)$ the components of the  dissipative vector field are generated via
\bq
Y_S^i=\left.\frac{\p \Xi(z,z_*)}{\p z_{*i}}\right|_{z_*=\p S/\p  z}\,.
\label{GVF}
\eq
Thus,  in the special  case where
\bq
\Xi(z,z_*)= \frac12 \, \frac{\p S}{\p z^i}\, G^{ij}(z) \, \frac{\p S}{\p z^j}\,, 
\eq
the dissipative  vector field is
\bq
Y_S^i=G^{ij}(z)\frac{\p S}{\p z^j}\,,
\eq
which is equivalent to that generated by a metriplectic 2-bracket as originally proposed in \cite{pjm84,pjm84b,pjm86}. Thus, it   has  been said that metriplectic dynamics is a special case of GENERIC. We will show that this is not the case by showing how vector fields of the form of \eqref{GVF} can be generated  by  a metriplectic 2-bracket. 
 
As before, we suppose our phase space is some finite-dimensional manifold $\mathcal{Z}$ on which lives a dynamical system, and suppose   smooth functions $f,g,H,S\in \La^0(\calz)$ with, as usual,  $H$ being the Hamiltonian and $S$ the entropy.  Then, in   coordinates, GENERIC dissipative dynamics is generated by 
\bq
(f,g))=\frac{\p  f}{\p z^i} Y_S^i(z,\p g/\p z)\,,
\label{nosym}
\eq
a bracket that is a linear derivation in the first slot but not in its second slot;  at this point  $Y_S^i$ is considered an arbitrary function of  its arguments.  Dissipative  dynamics is generated  with an entropy function as follows:
\bq
\dot{z}^i = (z^i,S))=  Y_S^i(z,\p S/\p z)
\eq
and energy conservation is assumed to be satisfied because 
\bq
\dot{H}= \frac{\p  H}{\p z^i} Y_S^i(z,\p S/\p z) = 0\,.
\label{dotHG}
\eq
 Finally, entropy production requires 
\bq
\dot{S}=(S,S))\geq 0\,,
\eq
a property built into $\bfY_s$ by $Y_S^i\, \p S/\p z^i\geq 0$. 

In the above we could identify
\bq
 Y_S^i(z,z_*)=\frac{\p \Xi(z,z_*)}{\p z_{*i}}\,,
 \eq
 but the linearization procedure does  not require the dissipative vector field to have  this  form in  terms  of a   dissipation  potential. 

Given that at the outset one has  in  mind a dynamical system with a particular entropy function  $S$,  one can turn  the bracket of \eqref{nosym} into a  bilinear form  by solving
\bq
\hat{G}^{ij}\,  \frac{\p S}{\p z^j}=Y_S^i(z,\p S/\p z)\,.
\label{identN}
\eq
If one  can solve for $\hat{G}$,  then the following  bracket, 
\bq
(f,g\rangle=\frac{\p f}{\p z^i} \, \hat{G}^{ij}\, \frac{\p g}{\p z^j}\,,
\label{Gblin}
\eq
will yield  the dissipative vector field of \eqref{identN}   in the form $Y_S^i=(z^i,S\rangle$. 
The bracket of \eqref{Gblin}  is clearly bilinear in  $f$  and $g$, but  symmetry  is not guaranteed. 

Since any  $\hat{G}^{ij}$ that satisfies \eqref{identN} will do, we  assume  the following direct product form:
\bq
\hat{G}^{ij}=Y_S^i(z,\p S/\p z)\, M^j(z,\p S/\p z) \,,
\label{NMassum}
\eq
which upon insertion   into  \eqref{identN}  yields
\bq
 Y_S^i\left( 1- M^j\frac{\p S}{\p z^j}\right)=0\,.
 \eq
Upon choosing
\bq
M^j={\de^{jk} \frac{\p S}{\p z_k}}\Big/{\frac{\p S}{\p z^l}\frac{\p S}{\p z_l}}\,,
\eq
where ${\de^{jk} {\p S}/{\p z^k}}= {\p S}/{\p z_j}$, we obtain
\bal
\hat{G}^{ij}(z)&=Y_S^i(z,\p S/\p z)M^j(\p S/\p z)
\ncr
&=Y_S^i \, \frac{\p S}{\p z_j} \Big/{\frac{\p S}{\p z^l}\frac{\p S}{\p z_l}}\,,
\label{GG}
\eal
where   now we  interpret $\hat{G}$ to  be  a given matrix function of the coordinated  $z$.  With this choice we have
\bq
(H,g\rangle= \frac{\p H}{\p z^i} Y_S^i(z,\p S/\p z)\, M^j(z,\p S/\p z)\frac{\p g}{\p z^j} =0 \,,
\nn
\eq
for  all functions  $g$,  which builds in  degeneracy in the first argument of  the bracket. 

Thus any GENERIC vector  field  \eqref{GVF}   can  be generated by the bilinear bracket of the form of  \eqref{Gblin}.  However,  in  general $\hat{G}$ is not symmetric, so we  do not yet have a metriplectic 2-bracket.   So next we show, given some assumptions, how to symmetrize $(f,g\rangle$.   Although above we considered only  the dissipative dynamics, consider the full  dynamics for  some observale $o\in \La^0(\calz)$ to take the form,  
\[
\dot{o} =  J(\mathbf{d}o, \mathbf{d}H) + \hat{G}(\mathbf{d}o, \mathbf{d}S)
\]
where $J$ is a Poisson bracket, $\hat{G}$ is a rank $2$-tensor,  such as that of \eqref{GG}, which we assume satisfies  
\[
\hat{G}(\mathbf{d}H, \cdot) = 0,
\]
and $S$ and $H$ are fixed and have never vanishing differentials.\\

\begin{lemma}
There exists a symmetric tensor ${G}$ such that 
\bqy
\dot{o} &=&  J(\mathbf{d}o, \mathbf{d}H) + \hat{G}(\mathbf{d}o, \mathbf{d}S)
\nonumber\\
&=&  J(\mathbf{d}o, \mathbf{d}H) +  {G}(\mathbf{d}o, \mathbf{d}S)
\eqy
and 
\bq
  {G}(\mathbf{d} f,\mathbf{d}g) =  {G}(\mathbf{d}g,\mathbf{d}f),
\eq
for all $f,g\in\La^0(\calz)$, provided the following criteria are met: 
\begin{itemize}
  \item[(i)]  $S$ is a fixed generator 
  \item[(ii)]  We do not care about the symmetry properties of $\hat{G}$. That is, we only care about $\hat{G}(\mathbf{d}o,\mathbf{d}S)$ not $\hat{G}(\mathbf{d}S,\mathbf{d}o) $. 
  \item[(iii)] $\hat{G}$ is unimportant except for the equations of motion it generates.
\end{itemize}
Thus, we can make an {\bf equivalent} dynamical system generated by a  symmetric bilinear  form  with symmetric tensor ${G}$.  
\end{lemma}

\begin{proof}
\\ 
Let $p \in \mathcal{Z}$. Since $\mathbf{d}S$ is nonvanishing, there exists a chart $(U_p,\{z^i_p\})$ containing $p$ such that $\mathbf{d}S = \mathbf{d}z_p^1$ \cite{Lee}. Since every manifold is Lindel\"of, we can choose a countable subcover of \{$U_p\}$ which we index in some manner as $\{U_n\}_{n \in \mathbb{N}}$. 
For every $n\in \mathbb{N}$, we define a new tensor $G_{U_n}$ on $(U_n, \{z^i_n\})$  as follows:
\bal
 &{G}_{U_n}^{ij}(z)  \coloneqq  \hat{G}^{ij}(z)\,,\qquad \mathrm{for}\  i\geq j
\ncr
 &{G}^{ij}_{U_n}(z)  \coloneqq  \hat{G}^{ji}(z)\,, \qquad \mathrm{for}\  i<j\,.
\eal
 
 We note these coordinate equalities only work our special coordinates and are not in general true for any other coordinates system. However, since each ${G}_{U_n}$ is symmetric in one coordinate system, it must be symmetric in all coordinate systems.

Since $\mathcal{Z}$ is a manifold, there exists a countable partition of unity $\{ \eta_\alpha \}$ subordinate to the open cover $\{U_n\}$. For each $\alpha$ there exists an $N$ such that $\text{supp}(\eta_\alpha) \subset U_N$. For every $\alpha$ we define $(U_\alpha, z_\alpha):= (U_N,{z_N}) $ and $G_\alpha := G_{U_N}$. Since $G_\alpha$ is defined on $\text{supp}(\eta_\alpha)$ we may define the global symmetric tensor 
\[
G = \sum_{\alpha}   {G}_{\alpha} \eta_\alpha.
\]
To see that this tensor is what we want, let $p \in \mathcal{Z}$ and let $V$ be a neighborhood of $p$ such that $V$ intersects $\text{supp}({\eta_\alpha}) $ for only a finite number of $\alpha$. For each $\beta$ such that $V$ intersects $\text{supp}({\eta_{\beta}}) $, we can use the fact that the coordinates ($U_{\beta}, \{z^{i}_\beta\}$) exist everywhere $\eta_\beta$ is supported,  to conclude that 
\bal
G(\mathbf{d}o, \mathbf{d}S)|_p &= \sum_\alpha  {G}_{\alpha}(\mathbf{d}o,\mathbf{d}S)\big|_p \eta_\alpha(p)
\\
& = \sum_\beta  {G}^{ij}_{\beta}\frac{\partial o}{\partial z^i_\beta}\bigg|_p\frac{\partial S}{\partial z^j_\beta}\bigg|_p \eta_\beta (p)
\ncr
 &= \sum_\beta  {G}^{i1}_{\beta}\frac{\partial o}{\partial z^i_\beta}\bigg|_p \eta_\beta (p) 
 \ncr
 &= \sum_\beta \hat{G}^{i1}\frac{\partial o}{\partial z^i_\beta}\bigg|_p \eta_\beta (p) 
 \ncr
 &= \sum_\beta \hat{G}^{ij}\frac{\partial o}{\partial z^i_\beta}\bigg|_p \frac{\partial S}{\partial z_\beta^j}\bigg|_p \eta_\beta (p) 
\ncr
&= \sum_\beta\hat{G}(\mathbf{d}o,\mathbf{d}S) \eta_\beta(p) = \hat{G}(\mathbf{d}o,\mathbf{d}S) \,.
\nonumber
\eal
\end{proof} 

\medskip
Care should be taken near the vanishing points of $\bfd S$ since rectification arguments fail at the vanishing points of covector fields.

\subsection{Finite-dimensional examples}
\label{ssec:finiteEX}

{In this subsection we discuss finite-dimensional systems of dimension three.  However, we note that 
it is easy to construct metriplectic 4-brackets for systems of arbitrary dimension.  For example, this can be done  by using the Lie algebra construction of Sec.\  \ref{sssec:LAMbkts}.  Moreover, one can begin with a Lie-Poisson system and construct extensions on $n$-tuples as in \cite{pjmT00} by direct product, semidirect product, etc.  (See \cite{updike} for the heavy top.) Here, for simplicity we restrict to 3 dimensions.
}

\subsubsection{Free rigid body}
\label{sssec:Frb}

The Euler's equations for the free rigid body  have a Hamiltonian structure in  terms of a Lie-Poisson bracket as discussed in Sec.~\ref{ssec:poisson}  (see Chap.\ 17 of \cite{sudarshan}).  For this 3-dimensional system  the coordinates are the three components of the angular momenta $(L^1,L^2,L^3)$ and the Lie algebra is $\mathfrak{so}(3)$. Thus, the Lie-Poisson bracket has the form of \eqref{LPbkt} with  the  coordinates $z^k$ being $L^k$ and  $c^{ij}_{\ k}=-\ep_{ijk}$. The Hamiltonian  and Casimir of the  system are given by
\bal
H&= \frac{(L^1)^2}{2I_1}+ \frac{(L^2)^2}{2I_2} + \frac{(L^3)^2}{2I_3}
\label{hrb}
\\
&\hspace{-2.3cm}\mathrm{and}\nonumber
\\
C&= (L^1)^2 + (L^3)^2+ (L^3)^2\,,
\label{angrb}
\eal
respectively.  Here the parameters $I_i$ are the principal moments of inertia.

In \cite{pjm86} the metriplectic 2-bracket was given for this system,  so as to create a system that removes (or adds) angular momentum $C$ of \eqref{angrb} while preserving the energy $H$ of \eqref{hrb} as it approaches an equilibrium of rotation about one of its principal axes.  With slight reformatting  the bracket of equation (31) of \cite{pjm86}  becomes the following:
\bal
 {(f,g)_H}&= {(f,H;g,H)}
\label{rigidbod}\\
&= -\la \left[
\frac{\p H}{\p L^k}\frac{\p H}{\p L^l} 
\left(\de^{ik} \de^{jl} -\de^{ij} \de^{lk} \right) 
\frac{\p f}{\p L^i}\frac{\p g}{\p L^j}
\right]\,.
\nonumber
\eal
Thus  we see easily that the rigid body metriplectic 4-bracket is of the form of  \eqref{so3} and is in fact the simple K-N construction of Sec.\ \ref{ssec:KN} with Euclidean  metric.

 {Our sectional curvature for entropy production generated by  \eqref{rigidbod} is
\bal
\dot{S}&= (S,H;S,H)
\nonumber\\
&=-\la\big((\nabla_L H \cdot \nabla_L  S)^2- |\nabla_L  H|^2  |\nabla_L  S|^2 \big)\geq 0\,,
\label{so3SeCu}
\eal
where $\nabla_L=\p/\p \bfL$ and the inequality follows for $\la>0$. We point out that the conventional sectional  curvature in Riemannian geometry is normalized by a denominator, 
$|\bfY|^2 |\bfX|^2  - (\bfX\cdot \bfY)^2$ for vectors $\bfX,\bfY \in \mathfrak{X}(\calz)$.   With this normalization in the present context  we would divide \eqref{so3SeCu} by $|\nabla_L  H|^2  |\nabla_L  S|^2 -  (\nabla_L H \cdot \nabla_L  S)^2$, arriving at the constant production rate $\dot S=\la$.  At the outset, we could  have defined the metriplectic 4-bracket with this normalization; however,   we chose not to do this in order to preserve multilinearity of the 4-bracket.  
 }

\subsubsection{Kirchhoff-Kida ellipse}
\label{sssec:KKe}

Kida  generalized  Kirchhoff's reduction of the 2-dimensional Euler equations  of fluid mechanics (see Sec.\ \ref{sssec:2+1}){, obtaining a reduction to a set of ordinary differential equations}.   {The reduced dynamical system describes  a constant patch of vorticity enclosed by a elliptical boundary, where as  time proceeds} the boundary remains an ellipse.  Like the free rigid body, this system is a 3-dimensional Lie-Poisson system, but instead of $\mathfrak{so}(3)$ it has the Lie algebra  $\mathfrak{sl}(2,1)$.  

In \cite{pjmMF97}   it was shown that quadratic moments of the vorticity, say  $\om(x,y)$  constitute a  subalgebra of the 2-dimensional Euler fluid Poisson bracket (see Eq.~\eqref{2DeulerBkt} below and \cite{pjm82}).  The  coordinates for  the Kirchhoff-Kida system, say $(z^1,z^2,z^3)$, are linearly related  to  the vorticity moments $(\int d^2x\,  \om x^2, \int d^2x \, \om y^2,\int d^2x\,  \om xy)$.  The Poisson tensor for this case  is
\bq
J= 
\begin{pmatrix}
0& z^3& -z^2\\
-z^3 & 0 &-z^1\\
z^2& z^1&0
\end{pmatrix}\,,
\label{JKK}
\eq
which has the associated Casimir invariant, 
\bq
C=(z^1)^2-(z^2)^2-(z^3)^2\,,
\label{hyper}
\eq
and this Casimir is a measure of the area of the Kirchhoff ellipse raised to the fourth power.  We refer the reader to  \cite{pjmF11}  for the Hamiltonian for this system, but note the level sets of $H$  are curved sheets with a symmetry direction, because they are independent of  {the coordinate} $z^2$.  Thus, orbits of this Hamiltonian system can be understood in terms of intersection of the sheets with the Casimir hyperboloid  {defined by \eqref{hyper}}, similarly to how the free rigid body can be understood in terms of the  {intersection of the } angular momentum sphere with  the Hamiltonian ellipsoid.  However, for the Kida case  one has three  classes of orbits,  corresponding to an elliptical patch rotating, librating, or stretching to infinite aspect ratio{, which are easily delineated by examining these intersections}. 

Again, using a  metriplectic 4-bracket of the form of \eqref{so3} gives our desired result.    This 4-bracket produces  dynamics that will either increase or decrease the  Casimir, implying growth or shrinkage of the area of the ellipse, while the energy is preserved.   This is essentially a finite-dimensional  version of the selective decay hypothesis (e.g.\ \cite{lieth}), and {is, in a sense,} dual to the double bracket dynamics of \cite{pjm05,pjmF11}, where  the Hamiltonian  is extremized at fixed Casimir (area).

\subsubsection{Other 3-dimensional systems}
\label{sssec:other3D}

The examples of Secs.\ \ref{sssec:KKe}  and \ref{sssec:Frb} are based on the 3-dimensional Lie algebras 
$\mathfrak{so}(3)$ and $\mathfrak{sl}(2,1)$, respectively.  According to the Bianchi classification, there are 9 real 3-dimensional Lie algebras, and one can construct finite-dimensional Lie-Poisson systems (see \cite{pjmY20}  for a listing) and then  {construct} metriplectic 4-brackets from each of these.  Many of the Lie-Poisson systems have physical realizations, e.g., Type IV was shown in \cite{pjmYT17} to underlie a simple model for the rattleback toy that defies the normal understading of chirality.  Also, outside of Lie-Poisson dynamics, there is a 3-dimensional system that describes the  invicid interaction of tilted  fluid vortex rings  \cite{pjmK23}, a system that diverges in finite time.  A natural energy   conserving metriplectic 4-bracket can easily be constructed for this system as well.

\section{Infinite-dimensional metriplectic 4-brackets: field theories}
\label{sec:infinite}

\subsection{General Hamiltonian  and metriplectic field theories}
\label{ssec:field4bkts}

Here we briefly  review some general properties of brackets for field  theories. For further development see, e.g., \cite{pjm98,pjmF11} and in a somewhat more mathematical setting \cite{pjmBGS23}.

For field theory, we replace a discrete index $i$ of finite-dimensional theories with labeling by a continuous variable $z$ and a field component index $i$, with the  degrees of freedom  denoted by  $\chi$, a multicomponent field. The functions on phase space are replaced by functionals of the dynamical degrees of freedom which are maps of $\chi\mapsto\R$. More specifically, we  consider the dynamics of classical field  theories with multi-component fields 
\[
\chi(z,t)=\big(\chi^1(z,t), \chi^2(z,t), \dots, \chi^M(z,t)\big)
\]
 defined on $z\in \cald$ for times $t\in\R$, i.e., $\chi\colon 
\cald \times\R\rightarrow\R$.   Here we use $z$ to be a label space coordinate unlike in the prevous section where it  was a dynamical variable or phase  space coordinate.  In fluid  mechanics  $\cald$ would be the 3-dimensional domain occupied by the fluid and $z$ the coordinates of  this point. We will assume in general $\cald$ has dimension $N$. 

The time rate of change of functionals, maps from fields to real numbers,   will be generated by making use of various brackets,  and these involve a notion of functional or variational derivative.  These are  defined via the first variation, which for a functional $F$ is
\bq
\de F[\chi;\eta]=\left.\frac{d}{d\epsilon}F[\chi+\epsilon\eta]\right|_{\epsilon=0} 
= \int_{\cald} d^N\!z\, \frac{\de F[\chi]}{\de \chi^i}\eta^i \,,
\label{funcder}
\eq
where again repeated indices are to be summed.  Here  $\de F[\chi;\eta]$ is the Fr\'{e}chet derivative acting on $\eta(z)$ with the integral over $z$ acting as the pairing between the quantity  $\de F/\de \chi$ (the gradient) and $\eta(z)$ (the displacement). The function that is the evaluation of $\chi$ at a point $\hat{z}$ satisfies $\de \chi(\hat{z})/\de \chi (z)= \de(\hat{z}-z)$, with $\de$ being the Dirac delta function. 

 With this notation, a  general noncanonical Poisson bracket is a binary operator on functionals, say $F$  and $G$,  of the form
\bq
\{F,G\}=\!\!\int_{\cald}\! \!d^N\!z' \! \!\int_{\cald}\! \!d^N\!z'' 
\frac{\de F[\chi]}{\de \chi^i(z')} \calj^{ij}(z',z'')\frac{\de G[\chi]}{\de \chi^j(z'')}\,,
\label{genNCbkt}
\eq
Here  $\calj$ is the Poisson operator (replacing the Poisson tensor of Sec.\ \ref{ssec:poisson}) that must ensure that the Poisson bracket satisfies: antisymmetry, $\{F,G\}=-\{G,F\}$, and the Jacobi identity,  
$\{\{F,G\},H\} +\{\{G,H\},F\} +\{\{H,F\},G\} =0$,  for all functionals $F,G,H$.  As before, this   form builds in  bilinearity and the Leibnitz derivation properties. 

For  the  dynamics on infinite-dimensional Poisson manifolds $\calj$ is degenerate, as was the case in finite dimensions.  When it is degenerate and the nontrivial null space gives rise to the   Casimir invariants, $C$,  that satisfy $\{C,G\}=0$ for all functionals $G$. Thus all  is formally, if not rigorously, equivalent to  the finite-dimensional development of Sec.\ \ref{sec:finite}. (See e.g.\  \cite{pjm98,pjm05} for review.)

General symmetric brackets  were given in \cite{pjmF11}, ones that  can generate double bracket or metriplectic 2-bracket dynamics, 
\bq
(F,G)=\!\int_{\cald}\! \!d^N\!z' \! \!\int_{\cald}\! d^N\!z''\, 
\frac{\de F[\chi]}{\de \chi^i(z')}\, \calg^{ij}(z',z'')\frac{\de G[\chi]}{\de \chi^j(z'')}\,, 
\label{gensymbkt}
\eq
where the {\it }metric operator $\calg$, analogous to the $G$-metric of Sec.\ \ref{sec:finite},  is chosen to ensure $(F,G)=(G,F)$ and to be semidefinite.  We may also want to build degeneracies into $\calg$ so that there exist distinguished functionals $D$ that satisfy $(D,G)= 0$ for all $G$.  
   
A specific  form of general of \eqref{gensymbkt} was given in \cite{pjmF11}, which  is generalization of the symmetric brackets given in previous works \cite{pjm86,vallis,carn90,shep90,pjm09}, 
\bq
(F,G)= \!\! \int_{\cald}\! \!d^N\!z' \! \!\int_{\cald}\! d^N\!z''\,  \{F,\chi^i(z')\} \calk_{ij}(z',z'') \{\chi^j(z''),G\} \,,
\label{symbkt}
\eq
with $\{F,G\}$ is any Poisson bracket and $\calk$ is a  symmetric kernel that can be chosen at will, e.g., to effect smoothing.  With this form, 
the  Casimir invariants of $ \{F,G\}$ will automatically be distinguished functionals $D$.

Proceeding we define a general form for the field-theoretic metriplectic 4-bracket by replacing the 4-tensor by a   4-tensor-functional with coordinate form given by the following integral kernel:
\bqy
&& \hat{R}^{ijkl}(z, z', z'', z''') [\chi(z)) ] 
\\
&&\hspace{1cm} = \hat{R}(\mathbf{d} \chi^i(z) ,\mathbf{d} \chi^j(z') , \mathbf{d} \chi^k(z'')  , \mathbf{d}\chi^l(z''') ) [ \chi({z}) ] \,.
\nonumber
\eqy
Formally identifying the functional derivative with the exterior derivative we are lead naturally to the following metriplectic 4-bracket on functionals:
{
\bqy
(F,G; K , N)\! &=&\! \! \int  \! \!  d^N\!z \!\! \int \! \! d^N\!z'\! \! \int \! \! d^N\!z''  \!\! \int \! \! d^N\!z'''  \hat{R}^{ijkl}(z,z'\!,z''\!,\!z''')
\nonumber\\
&& \hspace{-.2cm} \times \frac{\delta F}{\delta \chi^i(z)} \frac{\delta G}{\delta \chi^j(z')} \frac{\delta K}{\delta \chi^k(z'')} \frac{\delta N}{\delta \chi^l(z''')} \,.
\label{MPFT4bkt}
\eqy
}
with {properties} built in making it minimally metriplectic as discussed in Sec.\ \ref{sec:finite}.

Because we are dealing with field theory,  the quantity {$\hat{R}^{ijkl}(z,z',z'',z''')$} of \eqref{MPFT4bkt} should be  defined distributionaly and in general is an operator acting on the  functional derivatives.  In particular, we could  write $\hat{R}$ in terms of it's Fourier transform. For pseudo-differential operators 
{\bqy
&&\hat{R}^{ijkl}(z,z',z'',z''') =
\nonumber\\
&& \hspace{.8cm} \int \!\!d^N\!{p} \, e^{i p \cdot z}\!\! \int \!\!d^N\!{p'} \, e^{i p' \cdot z'}\!\! \int \!\!d^N\!{p''} \, e^{i p'' \cdot z''}\!\! \int \!\!d^N\!{p'''} \, e^{i p''' \cdot z'''}
\nonumber
\\&&  \hspace{4cm}\times \ \tilde{R}^{ijkl}(p,p',p'',p''') \,.
\eqy
}

Analogous  to Sec.\  \ref{sssec:Lie-M}, we say  $\hat{R}$ is a Lie-metriplectic  if $\hat{R}$ does not depend directly on the values of the field variable $\chi$, although it can depend on the label $z$. When such a dependence  is present, we can interpret it   as a location dependent ``curvature" in the {manifold} of functions.

\subsection{Reduction to special cases in infinite dimensions}
\label{ssec:recuctionsI}

Reductions of the field theoretic metriplectic 4-bracket of \eqref{MPFT4bkt} {follow} in the same manner as the finite-dimensional reductions of Sec.\ \ref{ssec:recuctionsI}.  The metric 2-bracket follows as expected, $(F,G)_H=(F,H;G,H)$, the  {K-M} bracket according to $[F,G]_S=(F,G;S,H)$, various double brackets follow  from $(F,G)_S=(F,S;G,S)$, etc.  In Secs.\ \ref{ssec:fluid} and \ref{ssec:KT} we will give many examples that demonstrate  these reductions.  Rather than  treating a general case of linearizing and symmetrizing a GENERIC bracket, in  Sec.\ \ref{sssec:GBoltz} we do so for the specific case of bracket for   the Boltzmann equation given in \cite{grm84}.

\subsection{Fluid-like examples}
\label{ssec:fluid}
 
\subsubsection{1 + 1 fluid-like theories}
\label{sssec:1+1}

Now consider the case were we have a single real-valued field variable depending on one space and one time independent variable, $u(x,t)$.  We will give three  examples of dissipation generated by a 4-bracket.   We do so by using a version of the K-N decomposition of Sec.\ \ref{ssec:KN}, where the  tensors $\si$ and $\mu$ are in this field-theoretic context replaced by symmetric operators $\Si$ and  $M$. 
Using these operators a field  theoretic version of the   {K-N} product gives a 4-bracket of the following form:
\bqy
 (F,K;G,N) &=& \int_\R dx\, W (\Sigma \KN M)(F_u,K_u,G_u,N_u)
\nonumber\\ 
&=& \int_\R dx\, W\Big(\Sigma(F_u,G_u) M(K_u,N_u)
 \nonumber\\ 
 &&\qqquad  - \Sigma(F_u,N_u) M(K_u,G_u)
  \nonumber \\ 
&&\qqquad + M(F_u,G_u) \Sigma(K_u,N_u)
  \nonumber \\ 
  &&\qqquad - M(F_u,N_u) \Sigma(K_u,G_u) 
  \Big)\,,
  \label{KN11}
\eqy
where $W$ is an arbitrary  weight, depending on $u$ and $x$,  that multiplies $\Sigma \KN M$ and for  convenience we define  $F_u=  \de F/\de u$. We are  free to choose $W$ without destroying the 4-bracket algebraic symmetries. We note, as we will show, the form of 4-bracket of \eqref{KN11} can be generalized in various ways  to higher dimensions of both the dependent and  independent variables. 

\bigskip
For our {\it first example} we assume  the following symmetric operators:
\bqy
\Sigma(F_u,G_u) &=& - \frac{d}{dx} \frac{\de F}{\de u}\,  \frac{d}{dx} \frac{\de G}{\de u} 
=-\p F_u\p G_u
\label{Sigma1}
\\
M(F_u, G_u) &=& \frac{\de F}{\de u} \, \frac{\de G}{\de u}= F_u G_u\,,
\label{M1}
\eqy
where again for convenience we simplify  the notation by defining $\p= {\p}/{\p x}$.
In addition  we assume $W=\nu$, some constant, and the Hamiltonian and Casimir are given by
\bq
H=\int_\R dx \,u \andq S=\frac12 \int_\R dx\,u^2\,.
\label{HS1}
\eq
Inserting these into \eqref{KN11} gives 
\bq
(F,G)_H= (F,H;G,H)=\nu\int_\R dx\, F_u \p^2 G_u\,,
\eq
which produces in an equation of motion
\bq
(u,S)_H=\nu\, \p^2u\,,
\eq
the  usual form for  viscous  dissipation of  a  1-dimensional  fluid.  

In light of the above result and the metriplectic formalism, it is natural to ask which Hamiltonian theory has the Hamiltonian and Casimir of \eqref{HS1}?  Although not so well known,  one can construct a 1 + 1 Poisson bracket that has any desired  Casimir.  To this end, consider
\bq
\{F,G\}=\int_\R dx\, h(u) \big(F_u\p G_u-G_u\p F_u\big)\,,
\label{hbkt}
\eq
where the function $h(u)$ is unspecified. Using a theorem of \cite{pjm82} it is easy to show that \eqref{hbkt} satisfies the Jacobi identity.  A bracket of the form of \eqref{hbkt} that has $\int_\R dx \,u^2/2$ as a Casimir must satisfy
\bq
\{F,C\}= 0 \ \forall F\quad \Rightarrow \quad  2 h \p C_u + C_u\p u=0\,,
\eq
which easily solved to yield $h=1/u^2$.  Ignoring the singularity, we proceed and obtain the Hamiltonian dynamics with Hamiltonian of \eqref{HS1}, viz. 
\bq
\frac{\p u}{\p t}=\{u,H\}= \p (u^{-2})\,.
\label{HS2}
\eq
Thus, our metriplectic  system of  this example is 
\bq
\frac{\p u}{\p t}=\{u,H\}= \p (u^{-2})+\nu \p^2 u\,.
\label{HS3}
\eq

As with the Harry Dym equation \cite{kruskal75}, the singularity of \eqref{HS2}  can be removed by coordinate changes.  For example, setting  $w=2/u^3$ takes \eqref{HS2} into 
\bq
\frac{\p w}{\p t}= - w \p w\,,
\eq
the inviscid Burger's equation.   The bracket of \eqref{hbkt} can  be  transformed via the chain rule  into many forms:  the form  where $h=u$ is the Lie-Poisson  form and the  form  where $h$ is constant  is  Gardner's bracket \cite{gardner}, 
\bq
\{F,G\}=\int_\R  dx\, G_u \p F_u\,.
\eq

\bigskip

Our {\it second example} uses Gardner's bracket with  the  Hamiltonian 
\bq
H = \frac12\int_\R dx\, \left(  \frac{{u^3}}{6} - \frac{(\p u)^2}{2}  + c \frac{u^2}{2} \right)\,, 
\label{kdvHam}
\eq
which together generate the Korteweg-De Vries equation  in a  frame boosted by speed $c$. Gardner's bracket has the Casimir
\bq
S=\int_\R dx\, u\,.
\eq
Thus we have all the ingredients needed to construct a metriplectic system with dissipation that conserves  \eqref{kdvHam}.  Using \eqref{KN11}, again  with \eqref{Sigma1} and \eqref{M1}, this  dissipation is generated  by
\bqy
(u,S)_H&=&(u,H;S,H)
\nonumber\\
&=&-\p(WH_u\p H_u) - W (\p  H_u)^2
\eqy
where 
\bq
H_u=  c u+ \frac{u^2}{2} +  \p^2 u\,, 
\eq
we  leave $W$ arbitrary,  and
\bq
\dot{S}= -\int_\R dx\, W \left( \p  H_u\right)^2\,.
\label{dotS11}
\eq
By design, the  righthand side of  \eqref{dotS11} vanishes when  evaluated on $a\,$sech$^2(\alpha x)$, the  boosted single soliton solution, with appropriate $a$ and $\al$.

\bigskip

In our {\it  third  example} of  this  subsection, our final example,   we  choose for $\Si$, 
\bqy
\Si(F_u, G_u)(x) &=&
 \p F_u(x)   \calh[G_u](x)  
 \nonumber\\
&& \qquad +  \p  G_u(x)   \calh[F_u](x) \,,
\eqy
where $\calh$ is  the Hilbert transform
\bq
\calh[u]=\frac1\pi\dashint_\R dx' \frac{u(x')}{x-x'}\,,
\eq
with $\dashint$ denoting  the Cauchy principal value integral.  (See e.g.\  \cite{king}.)
For $M$ we choose again that of \eqref{M1} 
and again we choose  the Hamiltonian and entropy of \eqref{HS1}. Note,  $\Si(F_u, H_u)=0$ for all functionals $F$ because $H_u=1$, $\p H_u=0$, and  $\calh[1]=0$.  Thus we obtain
\bqy
(F,G)_H&=& (F,H;G,H)=  \int_\R dx\, W\, \Sigma(F_u,G_u) 
\\
&=& \int_\R \!\! dx\, W \big( \p F_u   \calh[G_u]  
+  \p  G_u   \calh[F_u] \big)\,.
\nonumber
\eqy
Using the formal anti-self-adjoint property   of the Hilbert transform, i.e., 
\bq
\int_\R dx \, f \, \calh[g]= -\int_\R dx \, g\,  \calh[f]\,,
\eq
 assuming $W$  is constant,  and noting that $\p \calh[u]=\calh[\p u]$,  we obtain
\bqy
(u,S)_H&=&-W\left( \p \calh[u] + \calh[\p u]\right)
\nonumber\\
&=&-2W\, \calh[\p u]\,.
\eqy
Upon choosing $W=\al_1/(4\sqrt{2\pi})$ we see this is precisely Ott \& Sudan  dissipation  \cite{ott-sudan} proposed for modeling electron Landau damping in a fluid model. This form has been used extensively in the magnetic fusion literature, based on a latter paper \cite{hammett-perkins}.

\bigskip

In this section we have seen how a variety of dissipation mechanisms in 1 + 1 models can be generated by 4-brackets.  Indeed, there is considerable room  for   generalization, e.g., in our last example by replacing the opertors $\p$ and $\calh$ by any formally anti-self-adjoint  operators.   In \cite{pjmBR13} (see Sec.\ 4.4) a large family of dissipative structures were given in terms of multilinear forms with symmetries that build  in invariance of a chosen set of quantities.   As 4-brackets build in the invariance of $H$, we can extend  to other quantities generalizing the present framework.   Proceeding along these lines is beyond  the scope of the present paper.

\subsubsection{2  +  1 plasma and  fluid-like theories}
\label{sssec:2+1}

A large  class of 2  +  1  Hamiltonian fluid-like theories exist in the fluid mechanics and plasma physics literature. 
These include the 2-dimensional Euler equation for the dynamics of scalar vorticity and, for  example, generalizations including  quasigeostrophic dynamics of the potential vorticity which have a single scalar field defined on a some 2-dimensional domain, say with coordinates $(x,y)$.    Another example is  the 1-dimensional Vlasov-Poisson system of plasma physics, for which the domain is the  2-dimensional phase space with coordinates, say with $(x,v)$. These theories all have a noncanonical Poisson bracket with a Lie-Poisson bracket based on the Lie-algebra realization on functions, (see  \cite{pjm80,pjm82,pjm98}), often called the symplectomorphism algebra, 
\bq
[f,g] =\frac{\p f}{\p  x} \frac{\p g}{\p  y} -\frac{\p f}{\p  y} \frac{\p g}{\p  x} \,,
\eq
with the infinite-dimensional Lie-Poisson bracket being
\bq
\{F,G\}= \int d^2x\,  \om [F_\om,G_\om]\,,
\label{2DeulerBkt}
\eq
where we use  the shorthand $F_\om=\de F/\de \om$.

Quite naturally the infinite-dimensional metriplectic  4-bracket akin to the finite-dimensional 4-bracket of  \eqref{LA4bkt} is the following:
 \bal
 (F,K;G,N)&=\int\! d^2x \int \!d^2x' \, \calg(\bfx,\bfx')
 \ncr
 &\hspace{1cm} \times [F_\om,K_\om](\bfx)\,  [G_\om,N_\om](\bfx')\,,
 \eal
 which for symmetric $\calg(\bfx,\bfx')$ has the  minimal metriplectic symmetries. 
In the special case where $\calg=\la \de(\bfx-\bfx')$ with  $\la\in\R$, this  reduces to 
\bq
 (F,K;G,N)=\la \int \!d^2x \,  [F_\om,K_\om] [G_\om,N_\om]\,.
 \label{1+24bkt}
 \eq 

If we insert the enstrophy 
\bq
S=\frac12\int d^2x\, \om^2
\eq
 into \eqref{1+24bkt} as follows, we obtain:
\bal
 (F,S;G,S)&=\la \int \!d^2x \,  [F_\om,S] [G_\om,S]
 \ncr
 &=\la\int \!d^2x \,  [F_\om,\om] [G_\om,\om]
  \ncr
 &=\la \int \!d^2x \,  F_\om[\om,[G_\om,\om]]
 \eal
Next, with the   2-dimensional Euler Hamiltonian 
\bq
H=\frac12\int d^2x\, \om\psi 
\eq
where the stream function  $\psi$ satisfies $\nabla^2\psi =\om$ and $H_\om=\psi$, 
we obtain
\bq
\frac{\p \om}{\p t}=  (\om, S;H,S)= -\la [\om,[\om,\psi]]\,,
\eq
which gives  the   double bracket dynamics first proposed in \cite{vallis}, which was generalized  and used extensively in a variety of contexts in \cite{pjm05,pjmF11,pjmF17,pjmFWI18,pjmFS19,pjmF22}.  In light of the development of Sec.\ \ref{sssec:double},  this was to be expected. 

Next, it is natural to ask what is the metriplectic 2-bracket that results from \eqref{1+24bkt}.  We find
\bal
(F,G)_H&=(F,H;G,H)
\ncr
&=\la \int \!d^2x \,  [F_\om,H] [G_\om,H]\,,
\eal
which is the metriplectic  2-bracket recorded   in   \cite{Gay-Balmaz2013} and   \cite{pjmBR13}. 
Extensive calculations using this  bracket appeared in the context  of 2-dimensional Euler flows and a generalization to magnetohydrodynamics  in the Ph.D.\  thesis of C.\ Bressan  \cite{cb22}.  Preliminary results were published  in  \cite{pjmBKM18} and the main results  will appear in  a paper under preparation  \cite{pjmBKM23}.   Our results  in these works reveal a caveat: Because of degeneracy, the system may not relax to what one expects!  This problem is  remedied by using a bracket based on  that given  in  Sec.\ \ref{sssec:landauCO} below.

\subsubsection{3 + 1 fluid-like theories}
\label{sssec:3+1}

Next  we consider two 3 + 1 fluid-like systems.  We choose our set of dynamical variables to be composed of densities,  $\chi=\{\rho, \si, \bfM\}$, where $\rho$ is the mass density, $\si$ is entropy per unit volume, and $\bfM=(M_1,M_2,M_3)$ is momentum density, all of which depend on $\bfx=(x_1,x_2,x_3)$,  a Cartesian coordinate.  In our first  example we choose our entropy Casimir to be the actual entropy of a fluid system, while for the second we choose the  helicity, which is a Casimir, to be  our entropy. 

\bigskip
 
In our {\it first example}  we desire a theory that conserves total mass, momentum and energy, with thermodynamics only depending on two thermodynamic variables, say $\rho$ and $\si$.  Generalizations where we include the chemical potential are possible, but we won't consider such now.  Thus we expect our 4-bracket to not depend on  functional  derivatives  with respect to $\rho$,  which  might produce density diffusion.  We  build a theory out of a 
K-N pair.  

The simplest choice imaginable for $M$ is given by 
\bq
 M(F_\chi,G_\chi)=F_\sigma G_\sigma  
\eq
where as before $F_\si=\de F/\de \si$.   Since $\Si$ will involve pairs of functional derivatives $F_{\bfM}=\de F/\de \bfM$ and, analogous to \eqref{M1}, derivatives so as to assure a diffusive nature, we are thus  led  the general  isotropic (invariant under rotations) Cartesian tensor of order 4, 
\bqy
\hat{\La}_{ikst}&=&\al \de_{ik}\de_{st} + \be (\de_{is}\de_{kt} +  \de_{it}\de_{ks}) 
\nonumber\\
&&\qqquad  +\  \ga (\de_{is}\de_{kt}-  \de_{it}\de_{ks})
\label{hatla}
\eqy
as an ingredient for creating  $\Si$, which might build-in Galilean symmetry.  Given the above we assume
\bq
\Si(F_\chi,G_\chi)=   \hat{\La}_{ijkl} \, \p_j F_{M_i}\p_k G_{M_l} + a \, \nabla F_\sigma \cdot \nabla G_\sigma\,.
\label{siga}
\eq
where $\p_i:= \p /\p x_i$, $F_{M_i}=\de F/\de M_i$,  and we assume  $\al,\be,\ga,a$ are arbitrary functions  the thermodynamics variables $\rho$ and $\si$.  Putting this all together in the 3 + 1 context we obtain
{\bqy
(F,K;G,N) &=& \int d^3{x} \, (\Sigma \KN M)(F_\chi,K_\chi,G_\chi,N_\chi)
\nonumber\\
&=& \int d^3{x} \, \big( \Sigma(F_\chi,  G_\chi) M(K_\chi,N_\chi)
   \label{KN31}\\ 
 &&\qqquad  - \Sigma(F_\chi, N_\chi) M(K_\chi,G_\chi)
  \nonumber \\ 
&&\qqquad + M(F_\chi,G_\chi) \Sigma(K_\chi,N_\chi)
  \nonumber \\ 
  &&\qqquad - M(F_\chi,N_\chi) \Sigma(K_\chi,G_\chi) 
  \big)\,. 
\nonumber
\eqy}
Now, choosing the parameters with a specific target in mind we pick
\bqy
\Si(F_\chi,G_\chi) &=&  \frac{({\xi-2\eta/3}) }{\lambda T} (\nabla\cdot F_\mathbf{M})(\nabla\cdot G_\mathbf{M}) 
\nonumber\\
&& +\  \frac{\eta}{\lambda T}\big({\partial_i F_{M_k}} {\partial_i G_{M_k}} +  {\partial_iF_{M_k}}  {\partial_k G_{M_i}}\big)
\nonumber\\
&&+ \frac{\kappa}{\lambda T^2}(\nabla F_\sigma \cdot \nabla G_\sigma) \,,
\eqy
{where  choices for  the  parameters  $\al,\be,\ga$ of \eqref{hatla} and $a$  of \eqref{siga}  have  been made,  giving the parameters  temperature $T$,   viscosities $\xi$  and $\eta$, and thermal conductivity $\ka$.  This} leads to the following complicated 4-bracket: 
\bqy
(F,K; G, N) &=&  
 \int \! \frac{d^3x}{T} \, \frac{({\xi-2\eta/3})}{\lambda} 
 \nonumber\\
&&\hspace{-2.2cm} \times\  \Big[K_\sigma \nabla \cdot F_{\mathbf{M}} - F_\sigma \nabla \cdot K_{\mathbf{M}}\Big] 
 \Big[N_\sigma \nabla \cdot G_{\mathbf{M}} -G_\sigma \nabla \cdot N_{\mathbf{M}}\Big]
\nonumber\\
&&\hspace{-1.9cm} +  \int \! \frac{d^3x}{T} \,   \frac{\eta}{\lambda}
\Big(F_\sigma G_\sigma \big({\partial_i K_{M_k}}{\partial_i N_{M_k}} + \,  {\partial_i K_{M_k}}{\partial_k N_{M_i}}\big) 
\nonumber\\
&&  \hspace{-.15cm}
 +\   K_\sigma N_\sigma \big({\partial_i G_{M_k}}{\partial_i F_{M_k}} + \,  {\partial_i G_{M_k}}{\partial_k F_{M_i}}\big) 
\nonumber\\
&& \hspace{-.15cm} -\ K_\sigma G_\sigma  \big({\partial_i F_{M_k}}{\partial_i N_{M_k}} + \,  {\partial_i F_{M_k}}{\partial_k N_{M_i}}\big) 
\nonumber\\
&&  \hspace{-.15cm}  -  \ F_\sigma N_\sigma \big({\partial_i G_{M_k}}{\partial_i K_{M_k}} + \,  {\partial_i G_{M_k}}{\partial_k K_{M_i}}\big) \! \Big)
\nonumber\\
&& \hspace{-.3cm} + \  \int \! \frac{d^3x}{T^2} \, \frac{\kappa}{\lambda} \Big(  
F_\sigma G_\sigma \big(\nabla K_\sigma \cdot \nabla H_\sigma\big)
\nonumber\\
&& \hspace{1.5cm}   + K_\sigma N_\sigma (\nabla F_\sigma \cdot \nabla G_\sigma)
\nonumber\\
&&\hspace{1.5cm}  - \ N_\sigma G_\sigma\big(\nabla F_\sigma \cdot \nabla H_\sigma\big) 
\nonumber\\
&& \hspace{1.5cm}   - F_\sigma N_\sigma \big(\nabla K_\sigma \cdot \nabla G_\sigma\big) \Big)\,.
\label{NS4bkt}
\eqy
With the ideal fluid Hamiltonian  
\bq
H=\int \!d^3x {\left(\frac{|\bfM|^2}{2\rho} + \rho U(\rho, s)\right)}
\label{fluidHam}
\eq
where $\bfM=\rho \bfv$ and $\si=\rho s$ with $s$ being the specific entropy and $\rho$ the  mass  density, 
the  4-bracket  of   \eqref{NS4bkt} yields the following metriplectic 2-bracket:
\bqy
(F,G)_H &=& (F,H;G,H)
\nonumber\\
&=&\frac{1}{\lambda} \int  \! d^3x\, T\Lambda_{i k m n} \!\left[\frac{\partial }{\partial x_i}\left( \frac{\delta F}{\delta M_k}\right) - \frac{1}{T } \frac{\partial v_i}{\partial x_k}\frac{\delta F}{\delta \sigma}\right] 
\nonumber\\
&&\hspace{1.85cm}  \times\  \left[\frac{\partial }{\partial x_m}\left(\frac{\delta G}{\delta M_n}\right) 
- \frac{1}{ T } \frac{\partial v_m}{\partial x_n}\frac{\delta G}{\delta \sigma}\right] \nonumber\\
& +&
\int \! d^3x\,  \kappa T^2 \frac{\partial }{\partial x_k}\left[\frac{1}{ T} \frac{\delta F}{\delta \sigma}\right]
 \frac{\partial }{\partial x_k}\left[\frac{1}{ T} \frac{\delta G}{\delta \sigma}\right]\,,
 \label{pjm84b}
\eqy
where 
\bq
\Lambda_{i k m n} = \eta (\delta_{ni}\delta_{mk} + \delta_{nk}\delta_{mi} - \frac{2}{3}\delta_{ik}\delta_{mn}) + \xi \delta_{ik}\delta_{mn}\,.
\eq
We have written out \eqref{pjm84b} without abreviations so it is easy to see it is  precisely the  metriplectic bracket first given in \cite{pjm84b}.  The dynamics generated  by this bracket follows upon inserting  the entropy functional
\bq
S[\si]=\int \! d^3x \, \si\,{;}
\eq
{accordingly  $(\bfM,S)_H$ produces  a kind of viscous dissipation,  while $(\si,S)_H$ }  gives an entropy  equation with  thermal conduction and  viscous heating.   Together with  the ideal fluid Hamiltonian bracket given in \cite{pjmG80}, the metriplectic  system so generated,  is a version of the Navier-Stokes equation that conserves the energy of  \eqref{fluidHam} while producing entropy; i.e., it  produces  a fluid  dynamical realization of the first and second laws  of  thermodynamics.  See  \cite{pjm84b}   and \cite{pjmC20}  for details.

\bigskip

In our {\it second example}  we choose the helicity
\bq
S[\bfv]=\int \!d^3x \, \bfv\cdot\nabla\times\bfv \,,
\eq
which   is known to be a  Casimir for the ideal barotropic fluid \cite{pjm98}, to be our entropy.  We  insert this  into the metriplectic 4-bracket of \eqref{NS4bkt} along with the Hamiltonian of \eqref{fluidHam}, to obtain  
{\bqy
(F,S)_H&=&(F,H,S,H) 
\nonumber\\
&=&\frac1{\lambda} \int\!d^3x \bigg({T({\xi-2\eta/3}) } (\nabla\frac{1}{\rho} \cdot (\nabla \times \mathbf{v}))
\,\nabla\cdot F_\mathbf{M}
\nonumber\\
&&  + {T \eta}\Big({\partial_i \big((\nabla \times \mathbf{v})_k/\rho\big)}  ({\partial_i F_{M_k}}+ {\partial_k F_{M_i}})
 \Big) 
\nonumber\\
&&-  F_\sigma \Big( {({\xi-2\eta/3})}(\nabla \frac{1}{\rho} \cdot(\nabla \times \mathbf{v})) \nabla\cdot \mathbf{v}
\nonumber\\
&&+  \, \eta\, \partial_i\big((\nabla \times \mathbf{v})_k/\rho\big) (\partial_i v_k + \partial_k v_i) \Big)\bigg).
\eqy

This bracket will make entropy, helicity, while conserving the energy $H$  of \eqref{fluidHam}.  This is an interesting system in  its own right, which will be further investigated elsewhere.

\subsection{Kinetic theory examples}
\label{ssec:KT}

\subsubsection{Landau-like Collision Operator}
\label{sssec:landauCO}

In \cite{pjm84,pjm86} the metriplectic 2-bracket for the Landau-Lenard-Balescu (LLB) collision operator was given, one  that  generated a gradient flow using the  standard entropy.    Here we show how this 2-bracket dynamics  comes from a metriplectic 4-bracket.   The basic variable of this theory is the phase space density $f(z,t)$ (the distribution function), where a 6-dimensional phase space coordinate is $z=(\bfx, \bfv)=(x_1,x_2,x_3,v_1,v_2,v_3)$, which is typically a point in $T^*Q$, where $Q$ is a configuration manifold. Here  we won't emphasize  the  geometry and think of this  as $\R^6$.  By $\int\! d^6z$ we will mean an integration of this phase space.  

For functionals defined on $f$, as before, we abbreviate ${\delta F}/{\delta f}=F_f$, and for convenience 
 we define for functions $w\colon\R^6\rightarrow \R$,  the operator $P$, 
\bq
P[w]_i = \frac{\partial w(z)}{\partial v_i}  - \frac{\partial w(z')}{\partial v'_i} 
\eq
which is  a linear operator mapping functions on some subset of $\mathbb{R}^6$ to functions from $\mathbb{R}^{12} \rightarrow \mathbb{R}^3$ via the expression.  Also,  we define $\mathbf{g} = \mathbf{v} - \mathbf{v}'$ 
and 
\bqy
\omega_{ij} &=& \frac{1}{|\mathbf{g}|^3}(|\mathbf{g}|^2\delta_{ij} - g_i g_j)\delta(\mathbf{x} - \mathbf{x'}) 
\label{omdef}\\
&=& \delta(\mathbf{x} - \mathbf{x'})\frac{\partial^2}{\partial v_i \partial v_j}|\mathbf{v} - \mathbf{v}'|
=\delta(\mathbf{x} - \mathbf{x'})\frac{\partial^2 |\mathbf{g}|}{\partial v_i \partial v_j} \,.
\nonumber
\eqy
From \eqref{omdef}  it follows that 
\bqy
\omega_{ij} (z,z') &=& \omega_{ji} (z,z')\,, 
\\
\omega_{ij} (z,z') &=& \omega_{ij}(z',z)\,,
\\
 g_i\;\omega^{ij} &=& 0\,.
\eqy
Given the above, we can write down the  metriplectic 2-bracket that produces the LLB collision operator, the bracket that  was given in \cite{pjm84,pjm86}, 
\bq
(F, G)_H = {\int \! d^6{z}\! \int  \!d^6{z'} \,} P\left[{F_f}\right]_i \, T^{ij} \, P\left[{G_f}\right]_j\,,
\label{L2bkt}
\eq
where 
\bq
T^{ij} = \frac{1}{2}f(z)f(z')\;\omega^{ij}(z,z')\,,
\eq
with $\om_{ij}$ given by \eqref{omdef}.  The symmetric metriplectic 2-bracket of \eqref{L2bkt} together with the Poisson bracket for the Vlasov-Poisson system \cite{pjm80,pjm82},
\bq
\{F,G\}=\int\! d^6z\, f\, [F_f,G_f]\,,
\label{VPbracket}
\eq
with
\bq
[f,g]:=  \frac1m\left(\frac{\p f}{\p \bfx}\cdot \frac{\p g}{\p \bfv}- \frac{\p g}{\p \bfx}\cdot \frac{\p f}{\p \bfv}
\right)\,, 
\eq
generates the collisional Vlasov-Poisson system.  {This follows}  if  the following Hamiltonian for the Vlasov Poisson system, 
\bqy
H[f] &=&\frac12 \int\!d^6z\,  v^2 f(z) 
\label{VPham}\\
&& \qquad  + \frac12   {\int \! d^6{z}\! \int  \!d^6{z'} \,} V(z,z') f(z)f(z') \,, 
\nonumber
\eqy
{where $V$ is the Coulomb interaction potential,} is inserted along with  an appropriate Casimir chosen from the  set of Casimirs of \eqref{VPbracket}, {viz.}\ $\int\!d^6z\,  \calc(f)$ with $\calc$ an arbitrary function of $f$.  Choosing the following, which is  proportional to the  physical entropy:
\bq
S[f] = \int\! d^6 z\,  f \log f\,,
\label{VPentropy}
\eq
we have the results of  \cite{pjm84,pjm86}, where the system is generated by
\bq
\frac{\p f}{\p t}= \{f,H\} + (f,S)_{H}= \{f,\calf\} + (f,\calf)_H\,,
\eq
where $\calf=H + S$.

Now we construct the metriplectic 4-bracket, which comes quite naturally upon using a generalization of the  4-bracket  of \eqref{KN11}.  Let $\calg(z,z')$ be any kernel and suppose $\Sigma$ and $M$ are symmetric (under the integral) maps from functions to functions of $z$ and $z'$. Given such $\Si$ and  $M$,  a  K-N product on functional derivatives  can be defined as follows:
\bqy
\left( \Sigma \KN M \right)&&(F_f, K_f, G_f, N_f)(z,z')  
\nonumber\\
&& = \Sigma(F_f, G_f)(z,z') \  M(K_f, N_f)(z,z') 
\nonumber\\
&& - \Sigma(F_f, N_f)(z,z')   M(K_f, G_f)(z,z')
\nonumber\\
&& + M(F_f, G_f)(z,z')   \Sigma(K_f, N_f)(z,z') 
\nonumber\\
&& - M(F_f, N_f)(z,z')  \Sigma(K_f, G_f)(z,z') \,,
\eqy
from which  we define a bracket on functionals by
\bqy
(F,K; G,N) &=&  {\int \! d^6{z}\! \int  \!d^6{z'} \,}  \calg(z,z') 
\label{KN11-2}\\
&&\times (\Sigma \KN M)(F_f, K_f, G_f, N_f)(z,z') \,.
\nonumber
\eqy
This form of 4-bracket can be generalized to higher dimensions of both 
the dependent and  independent variables.  

We find the metriplectic 4-bracket for the LLB collision  operator  has the following simple symmetric form:
\bqy
(F,K; G, N) &=&  {\int \! d^6{z}\! \int  \!d^6{z'} \,}  \calg(z,z') 
\label{LLB4bkt}\\
&& \hspace{-1cm} \times \   (\delta \KN \delta)_{ijkl}\, P\left[{F_f}\right]_i \! P\left[{K_f}\right]_j  P\left[{G_f}\right]_k 
 P\left[{N_f} \right]_l\,,
 \nonumber
\eqy
where
\bqy
(\delta \KN \delta)_{ijkl}&=& \de_{ik}\de_{jl} - \de_{il}\de_{jk}
+ \de_{jl}\de_{ik} - \de_{jk}\de_{il}
\nonumber\\
&=& 2(\de_{ik}\de_{jl} - \de_{il}\de_{jk}) 
\eqy
and 
\bq
 \calg(z,z')=   \frac{\delta(\mathbf{x} - \mathbf{x'})f(z)f(z')}{4|\mathbf{g}|^3}\,.
 \label{calg}
 \eq
Inserting the $H$ of \eqref{VPham} into the bracket of  \eqref{LLB4bkt}, we find $(F,G)_H=(F,H;G,H)$ is that given by \eqref{L2bkt}.

In general any metriplectic 2-bracket of the form of  \eqref{L2bkt}, with any $T$, using a   metric  $g$ and a never vanishing Hamiltonian $H$, we can always define a parent metriplectic 4-bracket using the formula
\[
(F,K;G,N) = \int \frac{1}{g(\mathbf{d}H,\mathbf{d}H)}(T\KN g)(\mathbf{d}F, \mathbf{d}K,\mathbf{d}G,\mathbf{d}N)
\]
that satisfies the  relation   
\[
(F,G)_H = (F, H; G, H)
\]
For the case of Landau, this bracket is given by 
\bqy
(F,K;G,N) &=& \int \frac{1}{|\mathbf{g}|^2} d\mathbf{z} \; d\mathbf{z'} \left[T \KN \delta\right]_{ijkl}
\\
& \times&P\left[\frac{\delta F}{\delta f}\right]_i \! \!P\left[\frac{\delta K}{\delta f}\right]_j \!\! P\left[\frac{\delta G}{\delta f}\right]_k  \!\! P\left[\frac{\delta N}{\delta f}\right]_l .
\nonumber
\eqy

One of the advantages of the 4-bracket formalism is it allows various forms of dissipation to be effortlessly created and interchanged.  For example, suppose we replace \eqref{calg} by
\bq
 \calg_M(z,z')=   \frac{\delta(\mathbf{x} - \mathbf{x'})M\big(f(z)\big)M\big(f(z')\big)}{4|\mathbf{g}|^3}
 \label{calgM}
 \eq
where $M$ is an arbitrary function of $f$. The  metriplectic 4-bracket thus defined, with this  kernel,  the Hamiltonian \eqref{VPham}, and entropy given by
\bq
S[f]=\int\!d^6z\, s(f)\,,
\eq
can be designed to relax to a desired stable equilibrium.  If we choose $Ms''=1$, then $(F,G)_H=(F,H;G,H)$ is the metriplectic 2-bracket of \cite{pjm86}, that yields a gradient flow that relaxes to the state determined  by
\bq
H_f= -s'(f)
\eq
A rigorous Lyapunov stability argument would require $s'$  monotonic  and suitable convexity of $s$.  

A special  case  of the above construction   {occurs for the choice}
\bqy
M(f)&=&f(1-f)
\\
s(f) &=& \big(f\ln f + (1-f)\ln(1-f)\big)\,.
\eqy
The metriplectic 2-bracket $(F,G)_H=(F,H;G,H)$ was shown i in  \cite{pjm86} to produce a collision operator given in 
\cite{kadomtsev}, that was designed to relax to a Fermi-Dirac-like equilibrium state proposed in \cite{lynden-bell}. 
 
As a final example of this subsection, we show how to  covert the metriplectic theory for the LLP collision operator into  the bracket theory K-M theory of \cite{pjmK82} discussed in Sec.\ \ref{sssec:KM}.  This is a theory generated by the Hamiltonian with an antisymmetric bracket.  Here we suppose  $S = \int f \log(f)$ and obtain
\bqy
[F,K]_S &=&  (F,K; S, H) 
\nonumber\\
&=&   {\int \! d^6{z}\! \int  \!d^6{z'} \,}  \frac{f(z)f(z')\delta(x-x')}{2|\mathbf{g}|^3}
\nonumber\\
&&\hspace{-.9cm}  \times \   \left(P\left[\frac{\delta F}{\delta f}\right] \times P\left[\frac{\delta S}{\delta f}\right]\right)\cdot  \left(P\left[\frac{\delta K}{\delta f}\right] \times P\left[\frac{\delta H}{\delta f}\right]\right) \nonumber \\ 
&=&    {\int \! d^6{z}\! \int  \!d^6{z'} \,}  \frac{\delta(x-x')}{2|\mathbf{g}|^3} 
\nonumber\\
&&\  \times\ \bigg( P\left[\frac{\delta F}{\delta f}\right] \times \Big(f(z') \nabla_{v}f(z)
\nonumber\\
&& \hspace{1.85cm}  -\,  f(z) \nabla_{v'}f(z')\Big)\cdot  P\left[\frac{\delta K}{\delta f}\right] \times \mathbf{g}\bigg)\nonumber\,.
\eqy

 \subsubsection{Symmetrizing  and linearizing GENERIC for Boltzmann}
 \label{sssec:GBoltz}
 
 {As a final example}  we  show how to symmetrize and linearize a bracket  {given} by Grmela  in \cite{grm84} for the Boltzmann collision operator.   Thus  showing  how this system  is a metriplectic system.  Then, we show how it can be obtained from a metriplectic 4-bracket. 
 
Grmela's  bracket is
\bqy
&&(A,S)^{Gr} \! = \frac{1}{4} \! \int\!  \!d^6 z_2' \! \int\!  \!d^6 z_1' \! \int\! \! d^6 z_2 \! \int\! \!d^6 z_1\,  W(z_1', z_2', z_1, z_2)
\nonumber\\
&&\qquad \times\    \left [A_f(z_1) + A_f (z_2) - A_f(z_1') - A_f(z_2') \right ] 
\label{grmela}\\
&&\qquad \times\  \left[ \exp\left( S_f (z_1') + S_f(z_2') \right) - \exp\left( S_f (z_1) + S_f(z_2)  \right) \right]\,,
\nonumber
\eqy
where $f(z)$ is  again the phase space  density and $W$ is an integral  kernel with the following symmetries: 
\bqy
W(z_1', z_2', z_1, z_2) &=& W(z_1, z_2, z_1', z_2')
 \nonumber\\
 &=& W(z_2, z_1, z_1', z_2')  \nonumber\\
 &=& W(z_1, z_2, z_2', z_1'). 
\eqy
Furthermore, $W(z_1, z_2, z_1', z_2')$ is assumed to vanish  unless the following conditions are met:
\begin{itemize}
    \item[(i)] $\mathbf{v}_1^2 + \mathbf{v}_2^2 = \mathbf{v}_1^{'2}  + \mathbf{v}_2 ^{2'}$
    \item[(ii)]   $\mathbf{v}_1 + \mathbf{v}_2 = \mathbf{v}_1  + \mathbf{v}_2 $
    \item[(iii)]  $\mathbf{x_1} = \mathbf{x_2} = \mathbf{x_1}' = \mathbf{x_2}' $\,.
\end{itemize}
Dynamics generated by the entropy functional, 
\bq
S[f] = \int\! d^6 \!z \, f \log f
\label{Bentropy}
\eq
in the bracket of \eqref{grmela}  {gives} according to  $(f,S)^{Gr}$  the Boltzmann equation.

While this bracket works for obtaining the correct equations of motion, it is not bilinear nor symmetric. To rectify both of these problems we define a new bracket 
\bqy
(F,S)&=& \frac{1}{2} \!
\int \!\!d^Nz_1 \!\int \!\!d^Nz_2 \!\int \!\!d^Nz_1' \!\int \! \!d^Nz_2' \,\calw(z_1, z_2 , z_1', z_2')
\nonumber\\
&\times& 
\left[  F_f(z_2') + F_f(z_1') - F_f(z_2) - F_f(z_1) \right] 
\nonumber\\
&& \times  \left [ S_f(z_2') +S_f(z_1') - S_f(z_2) -S_f(z_1) \right]
\label{MBoltz} 
\eqy
with a new kernel, 
\bq
\calw:= \frac{W(z_1, z_2 , z_1', z_2')}{\log(\frac{f(z_2')f(z_1')}{f(z_2)f(z_1)})}
 \; f(z_1) f(z_2)\,.
\eq
Notice that the bracket   of  \eqref{MBoltz} has the metriplectic properties of  being  bilinear, symmetric, and that it recovers the Boltzmann equation when $S$ is the entropy of \eqref{Bentropy}. Furthermore, it is appropriately degenerate, i.e., because of the  properties of $W$, it  follows that 
\[
(F,H) = 0 \,,
\]
where the Hamiltonian $H$ satisfies 
\[
H_f(z) = \frac{1}{2} \mathbf{v}^2 + V(\mathbf{x}). 
\]
This allows the Boltzmann-Vlasov kinetic equation to be put  {into} metriplectic form, with the dynamics generated by the free energy.

If we define the symmetric maps
\bqy
P[F_f, G_f] &=& \Big(F_f(z_2') + F_f(z_1')
\nonumber\\
&&\hspace{2cm}  - F_f(z_2) - F_f(z_1)\Big)
\nonumber\\
&\times & \Big( G_f(z_2') + G_f(z_1')
\\
&&\hspace{2cm}  - G_f(z_2) - G_f(z_1))
\nonumber
\eqy
and 
\[
G[F_f, G_f] = F_f(z_1) G_f(z_1) + F_f(z_2) G_f(z_2).
\]
and also the {regulated} integral kernel by 
\bqy
&&U(z_1, z_2, z_1, z_1' ,z _2') = \frac{W(z_1, z_2 , z_1', z_2')}{\log(\frac{f(z_2')f(z_1')}{f(z_2)f(z_1) } )+i0   }
\\
&& \times \left( \frac{1}{\frac{1}{2}\mathbf{v}_1^2 + V(\mathbf{x}_1) +\frac{1}{2}\mathbf{v}_2^2 + V(\mathbf{x}_2) +i0 } \right)^2  \! f(z_1) f(z_2)\,.
 \nonumber\eqy
We can define the corresponding $4$ bracket  {as}
\bqy
&&(F,K;G,N) = \frac{1}{2}  \! \int \! d^Nz_1 \!  \int \! d^Nz_2 \!  \int \!  d^Nz_1'  \!  \int \!  d^Nz_2'  
\nonumber\\
&&\qquad \times\  U(z_1, z_2, z_1, z_1' ,z _2') (P \KN I)(F_f,K_f;G_f,N_f) \,,
\nonumber
\eqy
{which gives the desired  result.}
 
 
\section{Conclusion}
\label{sec:conclusion}

The main contribution of this work is the idea of endowing a manifold, finite or infinite, with the metriplectic 4-bracket structure, a bracket like the Poisson bracket  but with slots for four functions and properties motivated by those of curvature tensors.  Dynamics, flows on the manifold, are generated by two phase space functions, a Hamiltonian/energy $H$ and a  {Casimir/entropy} $S$, in such a way as to conserve energy and produce entropy.  The formalism naturally mates noncanonical  Hamiltonian dynamics, whose set of Casimirs includes candidate entropies, with dissipative dynamics generated by the metriplectic 4-bracket. The formalism encompasses previous dissipative bracket formalisms as special cases, and has rich geometrical structure; in fact, there exists much structure that was not covered in the present paper that will be treated in a future work. 

Many avenues  for further develpment are apparent. For example, given a Lie-Poisson bracket, there are a  variety of theories  {based} on  Lie-algebra extensions, the original paper \cite{pjm80} and the nondissipative fluid model of Sec.\ \ref{sssec:3+1} being examples.   Many other  magnetofluid models for plasma dynamics follow this framework (see e.g.\ \cite{pjmT00}).  A thorough geometric analysis and classification in the  metriplectic 4-bracket framework remains to be done.  As is well-known, noncanonical Hamiltonian dynamics arises via reduction; e.g.,  for fluids this is embodied  in the mapping from Lagrangian to Eulerian variables, with Lagrangian variables having standard canonical form and Eulerian being Lie-Poisson.  Thus the question arises of what happens on the un-reduced level as metriplectic dynamics transpires.  This was investigated  in  \cite{pjmM18} and \cite{pjmC20}, but a thorough understanding of  how  metriplectic 4-bracket dynamics relates to un-reduced dynamics deserves attention.  Lastly, we mention that  a more complete understanding  of symmetry and  conservation  in the metriplectic 4-bracket context would be helpful;  e.g., in previous work \cite{pjm84,pjmBR13}  this was done by considering multilinear brackets of various types in order to maintain Casimir or other dynamical invariants. 

In closing we suggest two practical uses for the metriplectic 4-bracket formalism:  as an aide  or framework for model building and as a kind of structure to  {be preserved} for computation. 
 
Fundamental Hamiltonian theories, e.g., with microscopic interactions involving many degrees of freedom, tend to be difficult to analyze and extract predictions.  Consequently,  one resorts to model building.  Sometimes models are obtained by identifying small parameters and performing rigorous asymptotics from the fundamental theory, resulting in  reduced systems that contain both Hamiltonian and dissipative parts. Good asymptotics will  lead to systems that respect the laws of energy conservation and entropy production. Alternatively, often models are based on phenomenology, using some known or believed properties, constraints, etc.\  in order to produce a model with desired behavior.  In the course of such an endeavor, one should obtain a model with clearly identifiable dissipative and nondissipative parts. Upon setting the nondissipative parts to zero, the remaining part should be Hamiltonian with a conserved Hamiltonian having a clear physical  interpretation as energy.  Similarly, the complete system should respect the law of  entropy production in addition to energy conservation.  Although,  sometimes the amount of heat produced may be so small so as to neglect energy conservation on large scales, as is  the  case for  turbulence studies with the Navier-Stokes equations. However, such a model should come from a more complete model including  entropy dynamics  like that given in  Sec.\ \ref{sssec:3+1}.  So, our claim is that the metriplectic 4-bracket formalism serves as a kind of paradigm, akin to role the Hamiltonian  {or Lagrangian formalisms have} played for obtaining fundamental theories.  It provides a convenient framework for building models with good thermodynamic properties. The K-N product of  Sec.\ \ref{ssec:KN}, although not the only tool available, can be useful in this regard. 

 Lastly, we suggest that the metriplectic 4-bracket  formulation provides a new avenue for structure preserving numerics (see e.g.\ \cite{pjm17} for overview).  Just as symplectic integrators (see e.g.\ \cite{HLW06} ) preserve Hamiltonian form by time stepping with canonical transformations, Poisson integrators do the same while preserving Casimir leaves (e.g.\  \cite{pjmKKS17,pjmJO22}),  and  various dissipative  brackets have been used and proposed for a variety of numerical schemes.  For example, the original goal of the double bracket of \cite{vallis,carn90,shep90} and the improvements in \cite{pjmF11} were to calculate vortex states, while  additional  calculations of fluid and magnetofluid stationary states  were given in  \cite{pjmF17,pjmFWI18,pjmFS19,pjmF22}.  Already, the metriplectic 2-bracket formulation has been used or proposed for computation \cite{pjmBKM18,kraus,cb22,pjmBKM23}, while  some exploratory metriplectic  4-bracket computations have been done in  \cite{updike}.

\section*{Acknowledgment}
\noindent   PJM was supported by U.S. Dept.\ of Energy Contract \# DE-FG05-80ET-53088 and a Humboldt Foundation Research Award.  He would like to thank  {Naoki Sato and Azeddine Zaidni  for proofreading and commenting on  an early draft of this work} and also   Michal Pavelka and Miroslav Grmela for explaining to him the latest version of  {what they mean by} GENERIC.

\bibliographystyle{unsrt}





\begin{thebibliography}{10}

\bibitem{rayleigh77}
J.~W.~S. Rayleigh.
\newblock {\em Theory of Sound}.
\newblock Macmillan and Co.\ and New York, 2nd edition, 1894.

\bibitem{pjm82}
P.~J. Morrison.
\newblock Poisson brackets for fluids and plasmas.
\newblock {\em AIP Conf. Proc.}, 88:13--46, 1982.

\bibitem{pjm98}
P.~J. Morrison.
\newblock {H}amiltonian description of the ideal fluid.
\newblock {\em Rev.\ Mod.\ Phys.}, 70:467--521, 1998.

\bibitem{pjmK82}
A.~N. Kaufman and P.~J. Morrison.
\newblock {A}lgebraic structure of the plasma quasilinear equations.
\newblock {\em Phys. Lett. A}, 88:405--406, 1982.

\bibitem{pjmH84}
P.~J. Morrison and R.~D. Hazeltine.
\newblock {H}amiltonian formulation of reduced magnetohydrodynamics.
\newblock {\em Phys. Fluids}, 27:886--897, 1984.

\bibitem{kauf84}
A.~N. Kaufman.
\newblock Dissipative {H}amiltonian systems: A unifying principle.
\newblock {\em Phys.\ Lett.\ A}, 100:419--422, 1984.

\bibitem{pjm84}
P.~J. Morrison.
\newblock {B}racket formulation for irreversible classical fields.
\newblock {\em Phys. Lett. A}, 100:423--427, 1984.

\bibitem{pjm84b}
P.~J. Morrison.
\newblock Some observations regarding brackets and dissipation.
\newblock Technical Report {PAM}--228, {U}niversity of {C}alifornia at
  {B}erkeley, March 1984.

\bibitem{pjm86}
P.~J. Morrison.
\newblock {A} paradigm for joined {H}amiltonian and dissipative systems.
\newblock {\em Physica D}, 18:410--419, 1986.

\bibitem{brockett}
R.~W. Brockett.
\newblock Dynamical systems that sort lists and solve linear programming
  problems.
\newblock {\em Proc.\ IEEE}, 27:799, 1988.

\bibitem{vallis}
G.~K. Vallis, G.~Carnevale, and W.~R. Young.
\newblock {E}xtremal energy properties and construction of stable solutions of
  the {E}uler equations.
\newblock {\em J.\ Fluid Mech.}, 207:133--152, 1989.

\bibitem{carn90}
G.~Carnevale and G.~Vallis.
\newblock Pseudo-advection relaxation to stable states of inviscid
  two-dimensional fluids.
\newblock {\em J.\ Fluid Mech.}, 213:549--571, 1990.

\bibitem{shep90}
T.~G. Shepherd.
\newblock A general method for finding extremal states of {H}amiltonian
  dynamical systems, with applications to perfect fluids.
\newblock {\em J.\ Fluid Mech.}, 213:573--587, 1990.

\bibitem{pjmF11}
G.~R. Flierl and P.~J. Morrison.
\newblock {H}amiltonian-{D}irac simulated annealing: Application to the
  calculation of vortex states.
\newblock {\em Physica D}, 240:212--232, 2011.

\bibitem{ch58}
J.~W. Cahn and J.~E. Hilliard.
\newblock Free energy of a nonuniform system. {I}. interfacial free energy.
\newblock {\em J. Chem. Phys..}, 28:258--267, 1958.

\bibitem{otto01}
F.~Otto.
\newblock The geometry of dissipative evolution equations: the porous medium
  equation.
\newblock {\em Com. Partial Diff. Eqs.}, 26:101--174, 2001.

\bibitem{ham82}
R.~S. Hamilton.
\newblock Three-manifolds with positive {R}icci curvature.
\newblock {\em J. Diff. Geom.}, 17:255--306, 1982.

\bibitem{perelman02}
G.~Perelman.
\newblock The entropy formula for the {R}icci flow and its geometric
  applications, 2002.

\bibitem{whittaker37}
E.~T. Whittaker.
\newblock {\em A Treatise on the Analytical Dynamics of Particles and Rigid
  Bodies}.
\newblock Cambridge University Press, 4th edition, 1937.

\bibitem{nambu73}
Y.~Nambu.
\newblock Generalized {H}amiltonian dynamics.
\newblock {\em Phys. Rev. D}, 7:2405--2412, 1972.

\bibitem{pjmB91}
I.~Bialynicki-Birula and P.~J. Morrison.
\newblock Quantum mechanics as a generalization of {N}ambu dynamics to the
  {W}eyl-{W}igner formalism.
\newblock {\em Phys. Lett. A}, 158:453--457, 1991.

\bibitem{pjm09a}
P.~J. Morrison.
\newblock {T}houghts on brackets and dissipation: Old and new.
\newblock {\em J. Physics: Conf. Ser.}, 169:012006, 2009.

\bibitem{Materassi12}
M.~Materassi and E.~Tassi.
\newblock Metriplectic framework for dissipative magneto-hydrodynamics.
\newblock {\em Physica D: Nonlinear Phenomena}, 241:729 --734, 2012.

\bibitem{pjmBR13}
A.~M. Bloch, P.~J. Morrison, and T.~S. Ratiu.
\newblock Gradient flows in the normal and {K}aehler metrics and triple bracket
  generated metriplectic systems.
\newblock In A.~Johann and et~al., editors, {\em Recent Trends in Dynamical
  Systems}, volume~35, pages 371--415. Springer Proceedings in Mathematics and
  Statistics, 2013.

\bibitem{pjmM18}
M.~Materassi and P.~J. Morrison.
\newblock Metriplectic torque for rotation control of a rigid body.
\newblock {\em J. Cybernetics and Physics}, 7:78--86, 2015.

\bibitem{pjmC20}
B.~Coquinot and P.~J. Morrison.
\newblock A general metriplectic framework with application to dissipative
  extended magnetohydrodynamics.
\newblock {\em J. Plasma Phys.}, 86:835860302 (32pp), 2020.

\bibitem{grm84}
M.~Grmela.
\newblock Particle and bracket formulations of kinetic equations.
\newblock {\em Contemp. Math.}, 28:125--132, 1984.

\bibitem{og97}
M.~Grmela and H.~C. {\"{O}}ttinger.
\newblock Dynamics and thermodynamics of complex fluids. part i. development of
  a general formalism.
\newblock {\em Phys. Rev. E}, 56:6620--6632, 1997.

\bibitem{holonomy}
R.~L. Fernandes.
\newblock Connection in {P}oisson geometry: {H}olonomy and invariants.
\newblock {\em J. Diff. Geom.}, 54(2):303--365, 2000.

\bibitem{weinstein}
A.~Weinstein.
\newblock The local structure of {P}oisson manifolds.
\newblock {\em J. Diff. Geom}, 18:523--557, 1983.

\bibitem{weinsteinE}
A.~Weinstein.
\newblock Errata and addenda:the local structure of {P}oisson manifolds.
\newblock {\em J. Diff. Geom}, 22:255, 1985.

\bibitem{sudarshan}
E.~C.~G. Sudarshan and N.~Mukunda.
\newblock {\em Classical Dynamics: A Modern Perspective}.
\newblock Wiley, 1974.

\bibitem{pjmG80}
P.~J. Morrison and J.~M. Greene.
\newblock Noncanonical {H}amiltonian density formulation of hydrodynamics and
  ideal magnetohydrodynamics.
\newblock {\em Phys. Rev. Lett.}, 45:790--793, 1980.

\bibitem{pjm80}
P.~J. Morrison.
\newblock The {M}axwell-{V}lasov equations as a continuous {H}amiltonian
  system.
\newblock {\em Phys.\ Lett.\ A}, 80:383--386, 1980.

\bibitem{mielke}
A.~Mielke.
\newblock Formulation of thermoelastic dissipative material using {GENERIC}.
\newblock {\em Continuum Mech. Thermodyn.}, 23:233--256, 2011.

\bibitem{Lojasiewicz}
S.~Lojasiewicz.
\newblock Une propri\'et\'e topologique des sous ensembles analytiques r\'eels.
\newblock {\em Colloques internationaux du C.N.R.S 117. Les \'Equations aux
  D\'eriv\'ees Partielles}, 117:87--89, 1963.

\bibitem{Polyak}
B.~T. Polyak.
\newblock Minimum enstrophy vortices.
\newblock {\em Zh. vych. mat.}, 3:643--653, 1963.

\bibitem{pjmBKM23}
C.~Bressan, M.~Kraus, O.~Maj, and P.~J. Morrison.
\newblock Metriplectic relaxation to equilibria: magnetohydrodynamics via
  collision-like metric brackets.
\newblock {\em Under Preparation}, 2023.

\bibitem{MWreduction}
J.~Marsden and A.~Weinstein.
\newblock Reduction of symplectic manifolds.
\newblock {\em Repts. Math. Phys.}, 5:121--1320, 1974.

\bibitem{lanczos}
C.~Lanczos.
\newblock Some properties of the {R}iemann-{C}hristoffel curvature tensor.
\newblock In {\em Recent Developments in General Relativity}, page 313.
  Pergamon Press, New York, N.Y., 1962.

\bibitem{kulkarni}
R.~S. Kulkarni.
\newblock On the {B}ianchi identities.
\newblock {\em Math. Ann.}, 199:175--204, 1972.

\bibitem{nomizu}
K.~Nomizu.
\newblock On the decomposition of generalized curvature tensor fields.
\newblock In {\em Differential Geometry, papers in Honor of K.\ Yano}, pages
  335--345. Kinokuniya, Tokyo, 1972.

\bibitem{fiedler03}
B.~Fiedler.
\newblock Determination of the structure algebraic curvature tensors by means
  of {Y}oung symmetrizers.
\newblock {\em S{\'e}minair. Lotharingien de Combinatoire}, 48:{B}48d, 2003.

\bibitem{goldberg}
S.~I. Goldberg.
\newblock {\em Curvature and Homology}.
\newblock Dover Publication, New York, 1982.

\bibitem{koszul}
J.~L. Koszul.
\newblock Homologie et cohomologie des alg\`ebres de {L}ie.
\newblock {\em Bul. Soc. Math. France.}, 78:65--127, 1950.

\bibitem{alioune}
B.~Alioune, M.~Boucetta, and A.~Lessiad.
\newblock On {R}iemann-{P}oisson {L}ie groups.
\newblock {\em arXiv:1908.05060}, 2019.

\bibitem{boucetta03}
M.~Boucetta.
\newblock Riemann-{P}oisson manifolds and {K}\"ahler-{R}iemann foliations.
\newblock {\em C. R. Acad. Sci. Paris, Ser. {I}}, 336:423--428, 2003.

\bibitem{milnor}
J.~Milnor.
\newblock Curvatures of left invariant metrics on {L}ie groups.
\newblock {\em Adv. Math.}, 21:293--329, 1976.

\bibitem{pjmF17}
M.~Furukawa and P.~J. Morrison.
\newblock Simulated annealing for three-dimensional low-beta reduced {M}{H}{D}
  equilibria in cylindrical geometry.
\newblock {\em Plasma Phys. Control. Fusion}, 59:054001, 2017.

\bibitem{pjmFWI18}
M.~Furukawa, T.~Watanabe, P.~J. Morrison, and K.~Ichiguchi.
\newblock Calculation of large-aspect-ratio tokamak and toroidally-averaged
  stellarator equilibria of high-beta reduced magnetohydrodynamics via
  simulated annealing.
\newblock {\em Phys. Plasmas}, 25:082506, 2018.

\bibitem{pjmF22}
M.~Furukawa and P.~J. Morrison.
\newblock Stability analysis via simulated annealing and accelerated
  relaxation.
\newblock {\em Phys. Plasmas}, 29:102504, 2022.

\bibitem{grm18}
M.~Grmela.
\newblock {G}{E}{N}{E}{R}{I}{C} guide to the multiscale dynamics and
  thermodynamics.
\newblock {\em J. Phys. Commun.}, 2:032001, 2018.

\bibitem{Lee}
J.~M. Lee.
\newblock {\em Introduction to Smooth Manifolds}.
\newblock Graduate Texts in Mathematics. Springer Verlag, New York, second
  edition, 2013.

\bibitem{pjmT00}
J-L. Thiffeault and P.~J. Morrison.
\newblock Classification of {C}asimir invariants of {L}ie-{P}oisson brackets.
\newblock {\em Physica D}, 136:397--414, 205--244.

\bibitem{updike}
M.~Updike.
\newblock Metriplectic heavy top: An example of geometrical dissipation.
\newblock Bachelor's thesis, University of Texas at Austin, 2022.

\bibitem{pjmMF97}
S.~P. Meacham, P.~J. Morrison, and G.~R. Flierl.
\newblock {H}amiltonian moment reduction for describing vortices in shear.
\newblock {\em Phys.\ Fluids}, 9:2310--2328, 1997.

\bibitem{lieth}
C.~E. Leith.
\newblock Minimum enstrophy vortices.
\newblock {\em Phys. Fluids}, 27:1388--1395, 1984.

\bibitem{pjm05}
P.~J. Morrison.
\newblock {H}amiltonian and action principle formulations of plasma physics.
\newblock {\em Phys.\ Plasmas}, 12:058102--1--13, 2005.

\bibitem{pjmY20}
Z.~Yoshida and P.~J. Morrison.
\newblock Deformation of {L}ie-{P}oisson algebras and chirality.
\newblock {\em J. Math. Phys.}, 61:082901, 2020.

\bibitem{pjmYT17}
Z.~Yoshida, T.~Tokieda, and P.~J. Morrison.
\newblock Rattleback: A model of how geometric singularity induces dynamic
  chirality.
\newblock {\em Phys. Lett. A}, 381:2772--2777, 2017.

\bibitem{pjmK23}
P.~J. Morrison and Y.~Kimura.
\newblock A {H}amiltonian description of finite-time singularity in {E}uler's
  fluid equations.
\newblock {\em Phys. Lett. A}, 484:129078, 2023.

\bibitem{pjmBGS23}
W.~Barham, Y.~G{\"u}{\c c}l{\"u}, P.~J. Morrison, and E.~Sonnendr{\"u}cker.
\newblock A self-consistent hamiltonian model of the ponderomotive force and
  its structure preserving discretization.
\newblock {\em arXiv:2309.16807v1 [physics.comp-ph]}, 2023.

\bibitem{pjm09}
P.~J. {Morrison}.
\newblock {On Hamiltonian and Action Principle Formulations of Plasma
  Dynamics}.
\newblock {\em AIP Conf. Proc.}, 1188:329--344, 2009.

\bibitem{kruskal75}
M.~Kruskal.
\newblock Nonlinear wave equations.
\newblock In J.~Moser, editor, {\em Dynamical Systems, Theory and
  Applications}, volume~38 of {\em Lecture Notes in Physics}, pages 310--354.
  Springer, Heidelberg, 1975.

\bibitem{gardner}
C.~S. Gardner.
\newblock {K}orteweg-de {V}ries equation and generalizations. iv. the
  {K}oretwegde. {V}ries equation as a {H}amiltonian system.
\newblock {\em J. Math. Phys.}, 12:1548--1551, 1971.

\bibitem{king}
F.~W. King.
\newblock {\em Hilbert Transforms}, volume 124 -- 125 of {\em Encyclopedia of
  Mathematics and its Applications}.
\newblock Cambridge University Press, Cambridge, 2009.

\bibitem{ott-sudan}
E.~Ott and R.~N. Sudan.
\newblock Nonlinear theory of ion acoustic waves with {L}andau damping.
\newblock {\em Phys. Fluids}, 12:2388--2394, 1969.

\bibitem{hammett-perkins}
G.~W. Hammett and R.~W. Perkins.
\newblock Fluid moment models for {L}andau damping with application to the
  ion-temperature-gradient instability.
\newblock {\em Phys. Rev. Lett.}, 64:3019--3022, 1990.

\bibitem{pjmFS19}
G.~R. Flierl, P.~J. Morrison, and R.~V. Swaminathan.
\newblock Jovian vortices and jets.
\newblock {\em Fluids: Topical Collection ``Geophysical Fluid Dynamics"},
  4:104, 2019.

\bibitem{Gay-Balmaz2013}
Francois Gay-Balmaz and Darryl~D Holm.
\newblock Selective decay by {C}asimir dissipation in inviscid fluids.
\newblock {\em Nonlinearity}, 26:495--524, 2013.

\bibitem{cb22}
C.~Bressan.
\newblock {\em Metriplectic relaxation for calculating equilibria: theory and
  structure-preserving discretization}.
\newblock PhD thesis, Technische Universit{\"a}t M{\"u}nchen, Zentrum
  Mathematik, Garching, Germany, 2022.

\bibitem{pjmBKM18}
C.~Bressan, M.~Kraus, P.~J. Morrison, and O.~Maj.
\newblock Relaxation to magnetohydrodynamic equilibria via collision brackets.
\newblock {\em J. Phys.: Conf. Series}, 1125:012002, 2018.

\bibitem{kadomtsev}
B.~B. Kadomtsev and O.~P. Pogutse.
\newblock Collisionless relaxation in systems with {C}oulomb interactions.
\newblock {\em Phys. Rev. Lett.}, 25:1155--1157, 1970.

\bibitem{lynden-bell}
D.~Lynden-Bell.
\newblock Statistical mechanics of violent relaxation in stellar systems.
\newblock {\em Mon. Not. Roy. Astron. Soc..}, 136:101--121, 1967.

\bibitem{pjm17}
P.~J. Morrison.
\newblock Structure and structure-preserving algorithms for plasma physics.
\newblock {\em Phys. Plasmas}, 24:055502, 2017.

\bibitem{HLW06}
E.~Hairer, C.~Lubich, and G.~Wanner.
\newblock {\em Geometric Numerical Integration}.
\newblock Springer Verlag, 2006.

\bibitem{pjmKKS17}
M.~Kraus, K.~Kormann, P.~J. Morrison, and E.~Sonnendr{\"u}cker.
\newblock {G}{E}{M}{P}{I}{C}: Geometric electromagnetic particle-in-cell
  methods.
\newblock {\em J. Plasma Phys.}, 83:905830401, 2017.

\bibitem{pjmJO22}
B.~Jayawardana, P.~J. Morrison, and T.~Ohsawa.
\newblock Clebsch canonization of {L}ie--{P}oisson systems.
\newblock {\em J. Geom. Mech.}, 2022.

\bibitem{kraus}
M.~Kraus and E.~Hirvijok.
\newblock Metriplectic integrators for the {L}andau collision operator.
\newblock {\em Phys. Plasmas}, 24:102311, 2017.

\end{thebibliography}


\end{document}